\documentclass{article}

\usepackage{amsmath}
\usepackage{amssymb}
\usepackage{amsthm}
\usepackage{authblk}
\usepackage[citestyle=alphabetic,bibstyle=alphabetic,backref=true]{biblatex}
\usepackage{braket}
\usepackage{complexity}
\usepackage[margin=1in]{geometry}
\usepackage{graphicx}
\usepackage{import}
\usepackage{mathtools}
\usepackage{subcaption}
\usepackage{thmtools}
\usepackage{titling}

\usepackage[unicode=true,bookmarks=true,hidelinks=false]{hyperref}
\usepackage[capitalise]{cleveref}

\newtheorem{theorem}{Theorem}
\newtheorem{fact}{Fact}
\newtheorem{remark}{Remark}
\newtheorem{lemma}{Lemma}
\newtheorem{corollary}{Corollary}

\newcommand{\numelec}{n}
\newcommand{\nummodes}{m}

\DeclarePairedDelimiter{\abs}{\lvert}{\rvert}
\DeclarePairedDelimiter{\norm}{\lVert}{\rVert}
\DeclarePairedDelimiterXPP\trace[1]{\Tr}{[}{]}{}{#1}

\DeclareMathOperator{\Tr}{Tr}
\DeclareMathOperator{\var}{var}

\DeclareMathOperator{\diag}{diag}
\DeclareMathOperator{\pfaff}{Pf}

\newcommand\expectation[2][]{\expectationover*{#1}{#2}}
\DeclarePairedDelimiterXPP\expectationover[2]{\mathbb E_{#1}}{[}{]}{}{#2}

\DeclareMathOperator{\perm}{Perm}
\DeclareMathOperator{\perfmatch}{PerfMatch}

\newcommand{\estimate}[1]{\hat{#1}}

\DeclareMathOperator{\bin}{bin}
\DeclareMathOperator{\seq}{seq}
\DeclareMathOperator{\binx}{binx}
\DeclareMathOperator{\seqx}{seqx}
\DeclareMathOperator{\seqxz}{seqxz}
\DeclareMathOperator{\sign}{sgn}

\author{Bryan O'Gorman
\thanks{IBM Quantum, IBM T.J. Watson Research Center, Yorktown Heights, NY, USA}
}

\title{Fermionic tomography and learning}

\hypersetup{
    pdftitle={\thetitle},
    pdfauthor={\theauthor},
}

\begin{document}
\begin{refsection}

\maketitle

\begin{abstract}
Shadow tomography via classical shadows is a state-of-the-art approach for estimating properties of a quantum state.
We present a simplified, combinatorial analysis of a recently proposed instantiation of this approach based on the ensemble of unitaries that are both fermionic Gaussian and Clifford.
Using this analysis, we derive a corrected expression for the variance of the estimator.
We then show how this leads to efficient estimation protocols for the fidelity with a pure fermionic Gaussian state (provably) and for 
an $X$-like operator of the form $(\ket{\mathbf 0}\bra{\psi} + \ket{\psi} \bra{\mathbf 0})$ (via numerical evidence).
We also construct much smaller ensembles of measurement bases that yield the exact same quantum channel, which may help with compilation.
We use these tools to show that an $\numelec$-electron, $\nummodes$-mode Slater determinant can be learned to within $\epsilon$ fidelity given $O(\numelec^2 \nummodes^7 \log(\nummodes / \delta) / \epsilon^2)$
samples of the Slater determinant.
\end{abstract}
 \section{Introduction}

Simulating quantum states of electrons is one of the most promising applications of quantum computing to chemistry and physics.
Fundamental to any quantum algorithm is a subroutine that extracts information from the quantum state.
Often this takes the form of estimating the expected values of some properties of the state.
Shadow tomography via classical shadows~\cite{huang2020predicting,elben2022randomized} is a recently proposed framework for doing so that uses randomized measurements to build a ``classical shadow'' from which the expected values of properties can be estimated.
In addition to the improved sample complexity, which is in some cases provably optimal, classical shadows also have the advantage that the measurements taken are independent of the specific properties to be estimated.

For many applications, everything of interest about the state is captured by the few-body reduced density matrices (RDMs).
Zhao et al.~\cite{zhao2021fermionic} applied the classical shadow approach to fermionic systems, showing how to additively estimate all $k$-RDMs of an $\nummodes$-mode state using $O\left(\binom{\nummodes}{k} k^{3/2} \log \nummodes\right)$ samples.
The ensemble of measurement bases they use is the set of all unitaries that are both affine (Clifford) and matchgate (fermionic Gaussian), which we accordingly call the affine-matchgate ensemble.
The channel defined by this affine-matchgate ensemble is the foundation of the present work.
We derive an expression for the second moment of a general operator in a general state, which in turn bounds the shadow norm of a general operator and thus the variance of the estimator thereof.

One such application is the learning of a Slater determinant, which is uniquely defined by its 1-RDMs.
We rigorously bound the sample complexity of doing so, translating the error in the estimated 1-RDMs into the fidelity between the learned and target Slater determinants.

\paragraph{Classical shadows}

The shadow tomography via classical shadows approach is to select an ensemble $\mathcal U$ of unitaries and measure in a basis $\hat{U}$ chosen uniformly at random from the ensemble. 
Doing so yields the channel
\begin{align}
\mathcal M_{\mathcal U}(\rho)
&= 
\expectation{
\hat{U}^{\dagger} \ket{\hat{\mathbf b}} \bra{\hat{\mathbf b}} \hat{U}
}
=
\frac{1}{\abs{\mathcal U}}
\sum_{U \in \mathcal U}
\sum_{\mathbf b \in {\{0, 1\}}^{\nummodes}}
\trace{\rho U^{\dagger} \ket{\mathbf b} \bra{\mathbf b} U}
{U}^{\dagger} \ket{{\mathbf b}} \bra{{\mathbf b}} {U}.
\end{align}
Importantly, the ensemble $\mathcal U$ is chosen so that the channel $\mathcal M_{\mathcal U}$ is invertible, in which case
\begin{align}
\rho &= 
\expectation{
\mathcal M^{-1}_{\mathcal U}
\left(
\hat{U}^{\dagger} \ket{\hat{\mathbf b}} \bra{\hat{\mathbf b}} \hat{U}
\right)
}.
\end{align}
Therefore, for any observable $O$, the estimator
\begin{align}
\hat{o} &= \trace{O \hat{\rho}},
&
\text{where}
&&
\hat{\rho} &= \mathcal M^{-1}_{\mathcal U} \left(\hat{U}^{\dagger} \ket{\estimate{\mathbf b}} \bra{\estimate{\mathbf b}} \estimate{U}\right),
\end{align}
is unbiased: $\expectation{\hat{o}} = \trace{O \rho}$.
For it to be useful, however, we must also bound its variance
\begin{align}
\var\left(\hat{o}\right) &= 
\expectation{
{\left(
\hat{o} - \trace{O \rho}
\right)}^2
}
\leq
\max_{\sigma: \text{state}}
\expectation{\hat{o}^2}
=
\norm{O}_{\mathcal U}^2,
\end{align}
where $\norm{O}_{\mathcal U}$, the \emph{shadow norm} of $O$, is the maximum root mean square of the estimator $\hat{o}$ over all normalized states.
Given a collection of observables, the sample mean over many classical shadows can be used to estimate their expectations with very favorable scaling, as captured by the following theorem.

\begin{theorem}[\cite{zhao2021fermionic,huang2020predicting}]
\label{thm:classical-shadows}
Suppose we have a collection of $L$ observables $O_1,\ldots, O_L$ .
Then, for all $\epsilon, \delta \in (0, 1)$, 
\begin{align}
M &= O\left(\frac{\log(2 L / \delta)}{\epsilon^2} \max_j \norm{O_j}_{\mathcal U}^2\right)
\end{align}
samples suffice to estimate $\trace{O_j\rho}$ for all $j$ with additive error $\epsilon$ and with probability at least $1 - \delta$.
Furthermore, if $\abs{\trace{\hat{o}_j}} \leq \norm{O_j}_{\mathcal U}^2$ for all $j \in [L]$ then sample-mean estimators suffice; otherwise, median-of-means estimators suffice.
\end{theorem}
Using the affine-matchgate ensemble, the restriction on $\abs{\trace{\estimate{o}_j}}$ is satisfied for $O_j = \Gamma_{\boldsymbol \mu}$ (as noted in \cite{zhao2021fermionic}, but not for, e.g., $O = \ket{\psi} \bra{\psi}$.

\paragraph{Fermionic states and operators}

We will describe fermionic systems in second quantization (i.e., the occupation basis).
Each computational basis state $\ket{\mathbf x}$ of an $\nummodes$-mode fermionic system is identified by an $\nummodes$-bit string $\mathbf x \in {\{0, 1\}}^{\nummodes}$, where $x_i = 1$ indicates that mode $i$ is occupied and $x_i = 0$ indicates that it is not.\footnote{The state $\ket{\mathbf x}$ may or may not be the same as the usual qubit basis state, depending on the fermion-to-qubit encoding used. 
(They are the same, for example, using the Jordan-Wigner encoding.)
Nothing in this paper depends on the particular fermion-to-qubit encoding.
}
A complete basis of operators is generated by the annihilation operators 
${\left(a_i\right)}_{i=1}^{\nummodes}$, defined by 
\begin{align}
a_i \ket{\ldots, x_{i-1}, 0, \ldots} &= 
0
&
a_i \ket{\ldots, x_{i-1}, 1, \ldots} &= 
{(-1)}^{\sum_{j < i} x_j}
\ket{\ldots, x_{i-1}, 0, \ldots},
\end{align}
and their Hermitian conjugates, the creation operators ${\left(a_i^{\dagger}\right)}_{i=1}^{\nummodes}$.
We also use the basis generated by the $2\nummodes$ single-mode Majorana operators
\begin{align}
\gamma_{2i-1} &= a_i + a_i^{\dagger}, 
&
\gamma_{2i} &= -i \left(a_i - a_i^{\dagger}\right).
\end{align}
For an index sequence $\boldsymbol \mu \in {[2\nummodes]}^k$, we denote by 
\begin{align}
\Gamma_{\boldsymbol \mu} &= {(-i)}^{\binom{k}{2}} \gamma_{\mu_1} \cdots \gamma_{\mu_k}
\end{align}
the corresponding degree-$k$ Majorana operator.
Furthermore, define
\begin{align}
\mathcal C_{2\nummodes, j}
&= \left\{\left(\mu_1, \ldots, \mu_j\right) : 1 \leq \mu_1 < \cdots \mu_j \leq 2\nummodes\right\}
&
\mathcal C_{2\nummodes, \mathrm{even}},
&=
\bigcup_{j=0}^{\nummodes} \mathcal C_{2\nummodes, 2j},
&
\mathcal C_{2\nummodes}
&=
\bigcup_{j=0}^{2\nummodes} \mathcal C_{2\nummodes, j},\\
\mathcal D_{2\nummodes, 2j}
&= 
\mathrlap{
\left\{
    \left(
    2{i_1} - 1, 2{i_1}, \cdots , 2{i_j} - 1, 2i_j 
    \right) : 1 \leq i_1 < \cdots < i_j  \leq \nummodes
\right\}
}
&&&
\mathcal D_{2\nummodes}
&=
\bigcup_{j=0}^{\nummodes} \mathcal C_{2\nummodes, 2j}.
\end{align}

\paragraph{Fermionic Gaussians and matchgates}
Fermionic Gaussian states and operators are special classes that are efficiently simulatable classically.
There are several different ways of defining matchgates and thus of showing their tractability~\cite{valiant2005holographic,jozsa2008matchgates,cai2009theory,bravyi2004lagrangian,terhal2002classical,knill2001fermionic}.
One way is to define a matchgate as a rank-$r$ tensor $T$ (with each leg having bond dimension $2$) with an associated $s \times s$ anti-symmetric matrix $A$, with $s \geq r$, such that for $\mathbf x \in {\{0, 1\}}^r$, the entry $T(\mathbf x)$ is given by the Pfaffian $\pfaff(A_{\mathbf x, \mathbf x})$ of the submatrix $A_{\mathbf x, \mathbf x}$, which excludes the $i$th row and $i$th column of $A$ when $x_i=0$.
(Incidentally, this correspondence with Pfaffians is why we use the term ``perfect matching'' for a partition of a set into pairs, despite the implicit graph always being the complete one.)
If such matchgate tensors are connected in a tensor network whose graph is planar (or, more generally, Pfaffian), then the tensor network can be contracted efficiently using the FKT algorithm\cite{kasteleyn1961statistics,temperley1961dimer}.

The correspondence between the physics and computer science terminologies is as follows.
An $\nummodes$-mode \emph{fermionic Gaussian unitary} is a unitary rank-$2r$ matchgate tensor.
An $\nummodes$-mode \emph{fermionic Gaussian state} is a normalized matchgate tensor, for both pure and mixed states.
An $\numelec$-electron, $\nummodes$-mode \emph{Slater determinant} is a pure fermionic Gaussian state with fixed particle number.

\paragraph{Fermionic partial tomography via classical shadows}

The Majorana operators 
$\left\{\Gamma_{\boldsymbol \mu} \middle| \boldsymbol \mu \in \mathcal C_{2\nummodes}\right\}$
form an orthonormal basis for Hermitian operators.
However, physical observables are spanned by the even-degree operators $\left\{\Gamma_{\boldsymbol \mu} \middle| \boldsymbol \mu \in \mathcal C_{2\nummodes, \mathrm{even}}\right\}$.
This is important because the affine-matchgate channel is invertible only on this even-degree subspace.
More specifically, physical observables are typically spanned by the low- and even-degree Majorana operators, so that their expectation values are completely determined by the $k$-body reduced density matrices ($k$-RDMs) for small $k$.
Zhao et al.\cite{zhao2021fermionic} showed that the $k$-RDMs can be efficiently estimated using classical shadows.
\begin{theorem}[\cite{zhao2021fermionic}]\label{thm:k-RDMs}
All $k$-RDMs can be estimated to within $\epsilon$ additive error and with probability at least $1 - \delta$ using 
$O\left(\binom{\nummodes}{k} k^{3/2} \log (\nummodes / \delta) / \epsilon^2\right)$ samples and in $\poly\left(\nummodes, \epsilon^{-1}, \log(1/\delta)\right)$ time.
\end{theorem}
Bonet-Monroig et al.\cite{bonet-monroig2020nearly} give essentially the same result for $k = 2$, using a deterministic set of measurement bases.
The analysis in \cite{bonet-monroig2020nearly} bears some similarity to that here; the edges of what we call a ``perfect matching'' are the generators of a set of commuting Majorana operators that they call a ``commuting clique''.

Our first result is an expression for the shadow norm of a general observable in an arbitrary state.
\begin{restatable}{theorem}{genvarthm}
\label{thm:moments}
For any state and any even-degree Hermitian
\begin{align}
\rho &= \sum_{\boldsymbol \tau \in \mathcal C_{2\nummodes}} g_{\boldsymbol \tau} \Gamma_{\boldsymbol \tau},
&
H &= 
\sum_{\boldsymbol \mu \in \mathcal C_{2\nummodes, \mathrm{even}}} h_{\boldsymbol \mu} \Gamma_{\boldsymbol \mu},
\end{align}
the second moment of the estimator $\hat{h} = \trace{H \hat{\rho}}$ is
\begin{align}
\expectation{\hat{h}^2} &= 
2^{\nummodes}
\sum_{
\boldsymbol \mu, \boldsymbol \mu' \in \mathcal C_{2\nummodes, \mathrm{even}}
: \text{$\abs{\boldsymbol \mu \cap \boldsymbol \mu'}$ even}
}
\kappa_{\nummodes, \boldsymbol \mu, \boldsymbol \mu'}
h_{\boldsymbol \mu} h_{\boldsymbol \mu'} g_{\boldsymbol \mu \boldsymbol \mu'},
\end{align}
where
\begin{align}
\kappa_{\nummodes, \boldsymbol \mu, \boldsymbol \mu'}
&=
\frac{
\lambda_{\nummodes, 
\frac{\abs{\boldsymbol \mu}}{2}, 
\frac{\abs{\boldsymbol \mu'}}{2}, 
\frac{\abs{\boldsymbol \mu \cap \boldsymbol \mu'}}{2}}
}{
\lambda_{\nummodes, \frac{\abs{\boldsymbol \mu}}{2}}
\lambda_{\nummodes, \frac{\abs{\boldsymbol \mu'}}{2}}
},
\\
\lambda_{\nummodes, \boldsymbol \mu}
&=
\lambda_{\nummodes, \frac{\abs{\boldsymbol \mu}}{2}}
=
\lambda_{\nummodes, k}
=
\begin{cases}
\frac{
\binom{\nummodes}{k}
}{
\binom{2\nummodes}{2k}
},  & k \in \mathbb Z,\\
0, &\text{otherwise},
\end{cases}
\\
\lambda_{\nummodes, \boldsymbol \mu, \boldsymbol \mu'}
&=
\lambda_{\nummodes, 
\frac{\abs{\boldsymbol \mu}}{2},
\frac{\abs{\boldsymbol \mu'}}{2}, 
\frac{\abs{\boldsymbol \mu \cap \boldsymbol \mu'}}{2}
}
=
\lambda_{\nummodes, k, k', a}
=
\begin{cases}
\frac{
\binom{\nummodes}{a, k-a, k'-a,\nummodes - k - k' + a}
}{
\binom{2\nummodes}{2a, 2(k-a), 2(k'-a),2(\nummodes - k - k' + a)}
}
, & k, k', a \in \mathbb Z,\\
0, &\text{otherwise}.
\end{cases}
\end{align}
\end{restatable}

From this we can derive the shadow norm for several properties of interest.

\begin{corollary}[Theorem 1 of \cite{zhao2021fermionic}]\label{cor:rdm-shadow-norm}
For every $\boldsymbol \mu \in \mathcal C_{2\nummodes, 2k}$,
$
\norm{\Gamma_{\boldsymbol \mu}}_{\mathcal U} 
=
\lambda_{\nummodes, k}^{-1}.
$
\end{corollary}
\noindent
\cref{cor:rdm-shadow-norm} follows from the fact that, for any state $\sigma = \sum_{\boldsymbol \tau} g_{\boldsymbol \tau} \Gamma_{\boldsymbol \tau}$,
\begin{align}
\expectation{\trace{\Gamma_{\boldsymbol \mu} \estimate{\sigma}}^2}
&=
2^{\nummodes}
\kappa^{-1}_{\nummodes, \boldsymbol \mu, \boldsymbol \mu} g_{\emptyset}
=
\frac{
\lambda_{\nummodes, \abs{\boldsymbol \mu}, \abs{\boldsymbol \mu}, \abs{\boldsymbol \mu}}
}{
\lambda_{\nummodes, \abs{\boldsymbol \mu}}^2
}
=
\lambda_{\nummodes, \abs{\boldsymbol \mu}}^{-1}.
\end{align}
\begin{corollary}\label{cor:shadow-norm-bound}
The shadow norm of a general even observable $H$ is upper bounded by 
\begin{align}
\norm{H}_{\mathcal U}^2
&\leq 
\sum_{\boldsymbol \mu, \boldsymbol \mu' \in \mathcal C_{2\nummodes, \mathrm{even}}
: \text{$\abs{\boldsymbol \mu \cap \boldsymbol \mu'}$ even}
}
\abs{h_{\boldsymbol \mu}}
\abs{h_{\boldsymbol \mu'}}
\kappa^{-1}_{\nummodes, \boldsymbol \mu, \boldsymbol \mu'}
.
\end{align}
\end{corollary}
\noindent
\cref{cor:shadow-norm-bound} follows because $\abs{g_{\boldsymbol \mu}} \leq 2^{-\nummodes}$ for any state and any $\boldsymbol \mu \in \mathcal C_{2\nummodes}$.

Importantly, it is not sufficient that the estimator have sufficiently small variance. 
It must also be efficiently calculable classically given the measurement unitary and outcome.
In general, and in particular for $H \in \left\{ \ket{\psi} \bra{\psi}, P(\psi, \theta)\right\}$, the estimator 
$\trace{\mathcal M^{-1} \left(H\right) U^{\dagger} \ket{\mathbf b} \bra{\mathbf b} U}$ does not correspond to a matchgate tensor network.
However, because the channel acts identically on each fixed-degree subspace, we can compute the estimator using a linear combination of $(\nummodes + 1)$
matchgate tensor networks, as detailed in \cref{sec:computation}.

\begin{restatable}{lemma}{computationlem}\label{lem:computation}
Let $\ket{\psi}$ and $\ket{\phi}$ be pure $\nummodes$-mode fermionic Gaussian states.
Then $\trace{\mathcal M_{\mathcal U}^{-1} \left( \ket{\psi} \bra{\psi}\right) \ket{\phi} \bra{\phi}}$  can be calculated exactly in $\poly(\nummodes)$ classical time.
If $\ket{\psi}$ is further restricted to having fixed, even particle number (i.e., to being a Slater determinant),
then $\trace{\mathcal M_{\mathcal U}^{-1} \left( \ket{\mathbf 0} \bra{\psi}\right) \ket{\phi} \bra{\phi}}$  can be calculated exactly in $\poly(\nummodes)$ classical time.
\end{restatable}

\begin{restatable}{lemma}{projectornormlem}\label{lem:projector-norm}
The shadow norm of the projector onto $\ket{\mathbf 0}$ is upper bounded by
\begin{align}
\norm{
{\left(\ket{0} \bra{0}\right)}^{\otimes \nummodes}
}_{\mathcal U} \leq \sqrt{2\nummodes}.
\end{align}
\end{restatable}
\noindent
\cref{lem:computation,lem:projector-norm} are proven in \cref{sec:computation,sec:shadow-norm-bounds}, respectively.

With the statistics of the estimators addressed by \cref{thm:classical-shadows} and their computability addressed by \cref{lem:computation}, we immediately get the following theorem.

\begin{theorem}\label{thm:projector-estimation}
Let $\rho$ be a normalized state and ${\left(\ket{\psi_i}\right)}_{i=1}^L$ a set of $L$ Slater determinants.
Then 
\begin{align}
\trace{\rho \ket{\psi_i} \bra{\psi_1}}, \ldots, \trace{\rho \ket{\psi_i} \bra{\psi_L}}
\end{align}
can be estimated with additive error $\epsilon$ and with probability at least $1 - \delta$ using 
\begin{align*}
O\left(s^2 \log(L / \delta) / \epsilon^2\right)
\end{align*}
samples
and $\poly(\nummodes, \epsilon^{-1}, \log(\delta^{-1}), s)$
classical processing time,
where $s = \norm{
{\left(\ket{0} \bra{0}\right)}^{\otimes \nummodes}
}_{\mathcal U}
$
Furthermore, 
\begin{align*}
\trace{\rho \ket{\mathbf 0} \bra{\psi_1}}, \ldots, \trace{\rho \ket{\mathbf 0} \bra{\psi_L}}
\end{align*}
can be estimated with additive error $\epsilon$ and with probability at least $1 - \delta$ with the same sample and time complexity but using 
$s = \norm{
{\left(\ket{0} \bra{1}\right)}^{\otimes \numelec}
\otimes
{\left(\ket{0} \bra{0}\right)}^{\otimes (\nummodes - \numelec)}
}_{\mathcal U}
$.
\end{theorem}

\begin{remark}\label{rmrk:projector-shadow-norm}
Numerically,  it appears that 
$\norm{
    {\left(
    \ket{0} \bra{1}
    \right)}^{\numelec}
    {\left(
    \ket{1} \bra{1}
    \right)}^{\nummodes - \numelec}
}_{\mathcal U}
\leq \frac12 \nummodes^c$ for $c < 1/2$.
See \cref{sec:numerics}.
\end{remark}

Estimating $X$-type observables of the form $\ket{\mathbf 0} \bra{\psi}$ for Slater determinant $\ket{\psi}$ was a major bottleneck 
in a recently proposed fermionic Monte Carlo algorithm\cite{huggins2021unbiasing}.
In the absence of a protocol to estimate $X$-type observables, they used classical shadows based on global Cliffords.
Doing so leads to tractable sample complexity, but requires computing the overlap of a Clifford state and a fermionic Gaussian state with inverse-exponential additive error, for which there is no known method (and which is probably \#P-hard).
\cref{thm:projector-estimation} and numerical evidence (\cref{rmrk:projector-shadow-norm}) suggest that this bottleneck can be overcome.

\paragraph{A smaller ensemble}
As shown in \cite{zhao2021fermionic}, the each unitary $U$ in the affine-Gaussian ensemble $\mathcal U$ corresponds to a permutation $\mathbf p \in \perm(2\nummodes)$ on the $2m$ single-mode Majorana operators.
We denote this correspondence by $U(\mathbf p)$, such that $U(\mathbf p) \Gamma_{\boldsymbol \mu} U(\mathbf p)^{-1} = \Gamma_{\mathbf p(\boldsymbol \mu)}$.
While using the full group $\perm(2\nummodes)$ of $(2\nummodes)!$ unitaries is useful for analysis, in practice a smaller ensemble of 
$\frac{(2\nummodes)!}{2^{\nummodes} \nummodes!}$ 
unitaries yields the exact same channel.
In \cite{huggins2021unbiasing}, the analogous fact for the global-Clifford ensemble was used to significantly reduce the cost of the measurement circuits.
We conjecture that something similar can be done for the affine-matchgate ensemble, reducing the circuit depth from $\sim \nummodes$ to $\sim \nummodes  / 2$.

The basic idea is that conditioned on a particular $\mathbf p \in \perm(2\nummodes)$, the channel that applies $U(\mathbf p)$, measures in the computational basis, and then applies $U(\mathbf p)^{\dagger}$ depends only on $\perfmatch(\mathbf p)$, where
\begin{align}
\perfmatch(\mathbf p) &= \left\{\left\{p_{2i - 1}, p_{2i}\right\} \middle| i \in [\nummodes]\right\}
\end{align}
is the perfect matching of $[2\nummodes]$ obtained by pairing up adjacent elements in the permutation.
Let $\perfmatch(2\nummodes)$ be the set of all perfect matchings of $[2\nummodes]$ (technically, of the complete graph with $2\nummodes$ vertices).
For each perfect matching $E \in \perfmatch(2\nummodes)$, there is the same number of permutations $\mathbf p \in \perm(2\nummodes)$ such that $\perfmatch(\mathbf p) = E$, and so it suffices to sample from $\perfmatch(2\nummodes)$ and for each one select a representative permutation, as captured by \cref{thm:smaller-ensembles}, whose proof is in \cref{sec:smaller-ensembles}.

\begin{theorem}\label{thm:smaller-ensembles}
Let $\mathcal P \subset \perm(2\nummodes)$  be a set of 
$\frac{(2\nummodes)!}{2^{\nummodes} \nummodes!}$ 
permutations such that 
\begin{align}
\left\{\perfmatch(\mathbf p) \middle| \mathbf p \in \mathcal P\right\} 
&= \perfmatch(2\nummodes).
\end{align}
Then the channel
\begin{align}
\mathcal M_{\mathcal P}(\rho) &= \expectationover*{
\mathbf p \in \mathcal P
}{
\sum_{\mathbf b \in {\{0, 1\}}^{\nummodes}}
\braket{ \mathbf b | U(\mathbf p)  \rho {U(\mathbf p)}^{\dagger}| \mathbf b}
{U(\mathbf p)}^{\dagger}    \ket{\mathbf b}\bra{\mathbf b} U(\mathbf p)
}
=
\mathcal M(\rho)
\end{align}
is the same as when using the full permutation group.
\end{theorem}

\paragraph{Learning a Slater determinant}

Aaronson and Grewal~\cite{aaronson2021efficient} attempted to learn a Slater determinant using only measurements in the computational basis.
In fact, that is not sufficient information.
Consider the two $\numelec$-electron, $2\numelec$-mode Slater determinants (in second quantization)
\begin{align}
\ket{\mathrm{bad},n,\pm} &= 2^{-n/2} {\left(\ket{01} \pm \ket{10}\right)}^{\otimes n}.
\label{eq:counterexample}
\end{align}
They are orthogonal,  $\braket{\mathrm{bad},n,+\mathrm{bad},n,-} = 0$, but have identical distributions when measured in the computational basis.

However, it is well-known that a Slater determinant is uniquely specified by its 2-RDMs.
In \cref{sec:slater-determinant}, We derive the following theorem, giving a rigorous, quantitative upper bound on the number of samples necessary to learn a Slater determinant.

\begin{restatable}{theorem}{sdlearningthm}
\label{thm:slater-determinant-learning}
Let $\ket{\psi}$ be an $\nummodes$-mode, $\numelec$-electron Slater determinant.
For any $\delta > 0 $ and any $\epsilon_{\mathrm{fid}} \in (0, \numelec / \nummodes]$,
there is a $\poly(m, \epsilon_{\mathrm{fid}}^{-1}, \delta^{-1})$-time quantum algorithm that 
consumes 
$O(\numelec^2 \nummodes^7 \log(\nummodes / \delta) / \epsilon_{\mathrm{fid}}^2)$
copies of $\ket{\psi}$
and
outputs a classical description of a Slater determinant $\ket{\tilde{\psi}}$ such that ${\left|\braket{\estimate{\psi} | \psi}\right|}^2 \geq 1 - \epsilon_{\mathrm{fid}}$ with probability at least $1 - \delta$.
\end{restatable}
The quantum component of the algorithm in \cref{thm:slater-determinant-learning} is extremely simple: independent measurements of the copies of the input state in a random affine-Gaussian basis.
The measurement outcomes are then processed completely classically.

 \section{Acknowledgements}
We thank 
Sergey Bravyi for suggesting the counterexample in \cref{eq:counterexample} and for helpful discussions,
Bill Huggins for pointing out the error in the general variance expression in \cite{zhao2021fermionic} and for helpful discussions,
and Kianna Wan for helpful discussions.
 
\appendix
\section{Statistics of the estimator}

\subsection{Additional notation}

It will help to introduce a small bit of notation.
For an index sequence $\boldsymbol \mu \in \mathcal D_{2\nummodes}$ corresponding to a diagonal Majorana operator $\Gamma_{\boldsymbol \mu}$, let $\bin(\boldsymbol \mu) \in {\{0, 1\}}^{\nummodes}$ be the corresponding bitstring.
We denote the inverse operation by $\seq(\mathbf x) \in \mathcal D_{2\nummodes}$. 
That is,
\begin{align}
\Gamma_{\boldsymbol \mu} &= \prod_{i=1}^{\nummodes} {\left(-i \gamma_{2i - 1} \gamma_{2i} \right)}^{\bin(\boldsymbol \mu)_i},
&
\Gamma_{\seq(\mathbf x)} &= \prod_{i=1}^{\nummodes} {\left(-i \gamma_{2i - 1} \gamma_{2i} \right)}^{x_i}.
\end{align}

Let 
\begin{align}
\mathcal D^*_{2\nummodes, 2j}
&=
\left\{
\boldsymbol \mu \in {[2m]}^{2j}
: \exists \mathbf p \in \perm(2\nummodes): \mathbf p(\boldsymbol \mu) \in \mathcal D_{2\nummodes, 2j}
\right\},
&
\mathcal D^*_{2\nummodes}
&=
\bigcup_{j=0}^{\nummodes}
\mathcal D^*_{2\nummodes, 2j}.
\end{align}
That is $\mathcal D^*_{2\nummodes}$ contains all $\boldsymbol \mu$ such that $\Gamma_{\boldsymbol \mu} = (-1)^{\binom{\abs{\boldsymbol \mu}}{2}} \gamma_{\mu_1} \cdots \gamma_{\mu_{\abs{\boldsymbol \mu}}}$ is diagonal, including those for which it is not the case that $\mu_1 < \cdots < \mu_{\abs{\boldsymbol \mu}}$.
If $\boldsymbol \mu \in \mathcal D^*_{2\nummodes}$, then there is a unique $\boldsymbol \mu' \in \mathcal D_{2\nummodes}$ such that $\Gamma_{\boldsymbol \mu} = \pm \Gamma_{\boldsymbol \mu'}$.
Define $\sign(\boldsymbol \mu) \in \{\pm 1\}$ to be the corresponding sign, i.e.,
\begin{align}
\Gamma_{\boldsymbol \mu} = 
\sign(\boldsymbol \mu) \Gamma_{\boldsymbol \mu'} = 
\sign(\boldsymbol \mu) \Gamma_{\seq(\bin(\boldsymbol \mu))}
.
\end{align}

There is another special subset of $\mathcal C_{2\nummodes}$ of interest.
For positive integer $j \leq \nummodes$, define
\begin{align}
\mathcal B_{2\nummodes, j}
&=
\bigtimes_{i=1}^{j} 
\left\{
2i - 1, 2i
\right\}.
\end{align}
As with $\mathcal D_{2\nummodes}$, each element of $\mathcal B_{2\nummodes, j}$ can be identified by a bitstring $\mathbf x \in {\{0, 1\}}^{j}$.
Analogous to $\bin$ and $\seq$, we define $\binx$ and $\seqx$ to be the functions that map between bitstrings and $\mathcal B_{2\nummodes, j}$:
\begin{align}
\Gamma_{\boldsymbol \mu} &=
{(-i)}^{\binom{j}{2}}
\prod_{i=1}^{j}
\gamma_{2i-1+{\binx(\boldsymbol \mu)}_i}
,
&
\Gamma_{\seqx(\mathbf x)} &=
{(-i)}^{\binom{j}{2}}
\prod_{i=1}^{j}
\gamma_{2i-1+x_i}
.
\end{align}

\begin{align}
\Gamma_{\seqx(\mathbf x)}
\Gamma_{\seqx(\mathbf y)}
&=
{(-1)}^{\binom{j}{2}}
\gamma_{2i-1+x_1}
\cdots
\gamma_{2i-1+x_j}
\gamma_{2i-1+y_1}
\cdots
\gamma_{2i-1+y_j}
\\
&=
\gamma_{2i-1+x_1}
\gamma_{2i-1+y_1}
\cdots
\gamma_{2i-1+x_j}
\gamma_{2i-1+y_j}
\\
&=
\prod_{i=1}^j
{(-1)}^{x_i(1-y_i)}
{\left(
\gamma_{2i-1}
\gamma_{2i}
\right)}^{x_i \oplus y_i}
=
i^{\norm{\mathbf y}_1 - \norm{\mathbf x}_1}
\Gamma_{\seq(\mathbf x \oplus \mathbf y)}
\end{align}
 \subsection{Useful facts}
We collect here a set of useful facts for reference.

\begin{fact}
For all integers $0 \leq k \leq m$,
\begin{align}
\binom{2\nummodes}{2k} &\leq 2^{\nummodes} \binom{\nummodes}{k}
\end{align}
\end{fact}

\begin{fact}
For all integers $0 \leq y \leq x \leq n$,
\begin{align}
\frac{
\binom{n+x}{k+y}
}{
\binom{2(n+x)}{2(k+y)}
}
&\leq 
\frac{
\binom{n}{k}
}{
\binom{2n}{2k}
}.
\end{align}
\end{fact}

\begin{fact}
For all integers $k \geq 0$,
\begin{align}
\binom{2k}{k}^{-1} &\leq \sqrt{k} 2^{1-2k}.
\end{align}
\end{fact}

\begin{fact}\label{fact:kappa-bound}
For all integers $0 \leq k, k' \leq \nummodes$ and $\max\{0, k  + k' - \nummodes\} \leq a \leq \min\{k, k'\}$,
\begin{align}
\kappa^{-1}_{\nummodes, k, k', a}
&\leq 
2^{\nummodes}
\frac{
\binom{k}{a}
}{
\binom{2k}{2a}
}
\frac{
\binom{\nummodes - k}{k' - a}
}{
\binom{2(\nummodes-k)}{2(k' - a)}
}
\end{align}
\end{fact}

 \subsection{Channel eigenvalues}

We'll start with a vastly simplified combinatorial proof that, for an arbitrary state
\begin{align}
\rho &= \sum_{\boldsymbol \tau \in \mathcal C_{2\nummodes}} g_{\boldsymbol \tau} \Gamma_{\boldsymbol \tau} ,
\end{align}
the channel is diagonal in the Majorana basis,
\begin{align}
\mathcal M_{\mathcal U} (\rho)
&= 
\sum_{\boldsymbol \tau \in \mathcal C_{2\nummodes}}
\lambda_{\nummodes, \boldsymbol \mu} 
g_{\boldsymbol \tau} \Gamma_{\boldsymbol \tau}
=
\sum_{\boldsymbol \tau \in \mathcal C_{2\nummodes, \mathrm{even}}}
\lambda_{\nummodes, \boldsymbol \mu} 
g_{\boldsymbol \tau} \Gamma_{\boldsymbol \tau},
\end{align}
with eigenvalues
\begin{align}
\lambda_{\nummodes, \boldsymbol \mu} 
&=
\lambda_{\nummodes, \abs{\boldsymbol \mu} / 2} 
=
\begin{cases}
\binom{\nummodes}{\abs{\boldsymbol \mu}/2}
/
\binom{2\nummodes}{\abs{\boldsymbol \mu}}
, & \text{$\abs{\boldsymbol \mu}$ even},\\
0, & \text{$\abs{\boldsymbol \mu}$ odd},
\end{cases}
\end{align}
which was originally shown in \cite{zhao2021fermionic} using the theory of finite frames\cite{han2000frames,waldron2018introduction}.
By linearity, it suffices to show that $\mathcal M_{\mathcal U}(\Gamma_{\boldsymbol \mu}) = \lambda_{\nummodes, \boldsymbol \mu}$ for every $\boldsymbol \mu \in \mathcal C_{2\nummodes}$.

\begin{align}
\mathcal M(\Gamma_{\boldsymbol \mu}) &= \expectationover*{
\mathbf p \in \perm(2\nummodes)
}{
\sum_{\mathbf b \in {\{0, 1\}}^{\nummodes}}
\braket{ \mathbf b | U(\mathbf p)  \Gamma_{\boldsymbol \mu} {U(\mathbf p)}^{\dagger}| \mathbf b}
{U(\mathbf p)}^{\dagger}    \ket{\mathbf b}\bra{\mathbf b} U(\mathbf p)
}
\\
&=
\frac{1}{(2\nummodes)!}
\sum_{\substack{
\mathbf p \in \perm(2\nummodes) \\
\mathbf b \in {\{0, 1\}}^{\nummodes}
}}
\braket{ \mathbf b | U(\mathbf p)  
\Gamma_{\boldsymbol \mu} 
{U(\mathbf p)}^{\dagger}| \mathbf b}
{U(\mathbf p)}^{\dagger}    
\left(
2^{-\nummodes}
\sum_{\mathbf x \in {\{0, 1\}}^{\nummodes}}
{(-1)}^{\mathbf b \cdot \mathbf x}
\Gamma_{\seq(\mathbf x)}
\right)
U(\mathbf p)
\label{eq:eigenvalue-expansion}
\\
&=
\frac{2^{-\nummodes}}{(2\nummodes)!}
\sum_{\substack{
\mathbf p \in \perm(2\nummodes)\\
\mathbf b, \mathbf x\in {\{0, 1\}}^{\nummodes}
}}
{(-1)}^{\mathbf b \cdot \mathbf x}
\braket{ \mathbf b | \Gamma_{\mathbf p(\boldsymbol \mu)} | \mathbf b}
\Gamma_{\mathbf p^{-1}(\seq(\mathbf x))}
\label{eq:eigenvalue-homomorphism}
\\
&=
\frac{2^{-\nummodes}}{(2\nummodes)!}
\sum_{\substack{
\mathbf p \in \perm(2\nummodes) : \mathbf p(\boldsymbol \mu) \in \mathcal D^*_{2\nummodes}\\
\mathbf b, \mathbf x\in {\{0, 1\}}^{\nummodes}
}}
{(-1)}^{\mathbf b \cdot (\mathbf x + \bin(\mathbf p(\boldsymbol \mu)))}
\sign(\mathbf p(\boldsymbol \mu))
\Gamma_{\mathbf p^{-1}(\seq(\mathbf x))}
\label{eq:eigenvalue-diagonal}
\\
&=
\frac{1}{(2\nummodes)!}
\sum_{
    \mathbf p \in \perm(2\nummodes) : \mathbf p(\boldsymbol \mu) \in \mathcal D^*_{2\nummodes}
}
\Gamma_{\boldsymbol \mu}
\label{eq:eigenvalue-delta}
\\
&=
\left(
\Pr_{\mathbf p \in \perm(2\nummodes)}
\left[
\mathbf p(\boldsymbol \mu) \in \mathcal D^*_{2\nummodes}
\right]
\right)
\Gamma_{\boldsymbol \mu}
=
\lambda_{\nummodes, \boldsymbol \mu}
\Gamma_{\boldsymbol \mu},
\label{eq:eigenvalue-prob}
\end{align}
where
\begin{align}
\lambda_{\nummodes, \boldsymbol \tau}
&
=
\Pr_{\mathbf p \in \perm(2\nummodes)}
\left[
    \mathbf p(\boldsymbol \tau) \in \mathcal D^*_{2\nummodes}
\right]
=
\frac{1}{(2\nummodes)!}
\binom{\nummodes}{\abs{\boldsymbol \mu} / 2}
\abs{\boldsymbol \mu}!
(2\nummodes - \abs{\boldsymbol \mu}))!
=
\binom{\nummodes}{\abs{\boldsymbol \mu} / 2}
/
\binom{2\nummodes}{\abs{\boldsymbol \mu}}.
\end{align}
In \cref{eq:eigenvalue-expansion}, we expand the projector $\ket{\mathbf b} \bra{\mathbf b}$ in the diagonal Majorana basis.
In \cref{eq:eigenvalue-homomorphism}, we make (twice) use of the group homomorphism.
In \cref{eq:eigenvalue-diagonal}, we make use of the fact that $\braket{\mathbf b | \Gamma_{\boldsymbol \tau} | \mathbf b}$ vanishes if $\boldsymbol \tau \notin \mathcal D_{2\nummodes}$ and is equal to ${(-1)}^{\mathbf b \cdot \bin(\boldsymbol \tau)}$ if $\boldsymbol \tau \in \mathcal D_{2\nummodes}$.
In \cref{eq:eigenvalue-delta}, we use the fact that $\sum_{\mathbf b} {(-1)}^{\mathbf b\cdot \mathbf y} = 2^{\nummodes} \delta_{\mathbf y}$, and specifically that
$\mathbf p^{-1} (\seq(\mathbf x)) = 
\mathbf p^{-1} (\seq(\bin(\mathbf p(\boldsymbol \mu)))) = 
\boldsymbol \mu$.
\cref{eq:eigenvalue-prob} makes clear why $\lambda_{\nummodes, \boldsymbol \mu}$ has the form that it does: it is simply the probability that a uniformly random permutation maps an index sequence of a certain length to a ``diagonal'' one.
 \subsection{Second moment of the estimator for general state and observable}\label{sec:general-statistics}

In this section, we prove \cref{thm:moments}, restated below.
\genvarthm*

As with $\lambda_{\nummodes, k}$, $\lambda_{\nummodes, \boldsymbol \mu, \boldsymbol \mu'} = \lambda_{\nummodes, k, k', a}$ has a combinatorial interpretation.
Specifically, it is the probability that under a uniformly random permutation $\mathbf p$ four disjoint index sequences 
$\boldsymbol \mu \setminus \boldsymbol \mu'$,
$\boldsymbol \mu' \setminus \boldsymbol \mu$,
$\boldsymbol \mu \cap \boldsymbol \mu'$, and
$(1, \ldots, 2\nummodes) \setminus (\boldsymbol \mu \cup \boldsymbol \mu')$
of respective sizes $2k$, $2k'$, $2a$, and $2(\nummodes - k- k' + a)$ are each simultaneously mapped to ``diagonal'' index sequences
$\mathbf p\left(\boldsymbol \mu \setminus \boldsymbol \mu'\right)$,
$\mathbf p\left(\boldsymbol \mu' \setminus \boldsymbol \mu\right)$,
$\mathbf p\left(\boldsymbol \mu \cap \boldsymbol \mu'\right)$, and
$\mathbf p\left((1, \ldots, 2\nummodes) \setminus (\boldsymbol \mu \cup \boldsymbol \mu')\right)$:
\begin{align}
\lambda_{\nummodes, \boldsymbol \mu, \boldsymbol \mu'}
&=
\lambda_{\nummodes, k, k', a}
\\
&=
\frac{1}{(2\nummodes)!}
\binom{\nummodes}{
k - a, k' - a, a, \nummodes - k - k' + a
}
(2(k-a))!
(2(k'-a))!
(2a)!
(2(\nummodes - k - k' + a))!
\\
&=
\frac{
\binom{
\nummodes
}{
k - a, k' - a, a, \nummodes - k - k' + a
}
}{
\binom{
2\nummodes
}{
2(k - a), 2(k' - a), 2a, 2(\nummodes - k - k' + a)
}
}.
\end{align}

\begin{proof}[Proof of~\cref{thm:moments}]
To begin, note that
\begin{align}
\hat{h}^2
&=
\trace{
H \mathcal M^{-1}_{\mathcal U} (\hat{\rho})
}^2
\\
&=
\trace{
\mathcal M^{-1}_{\mathcal U} (H) \hat{\rho}
}^2
\\
&=
{\left(
\sum_{\boldsymbol \mu \in \mathcal C_{2\nummodes, \mathrm{even}}}
\lambda_{\nummodes, \boldsymbol \mu}^{-1}
h_{\boldsymbol \mu}
\trace{
\Gamma_{\boldsymbol \mu} \hat{\rho}
}
\right)}^2
\\
&=
\sum_{\boldsymbol \mu, \boldsymbol \mu'\in \mathcal C_{2\nummodes, \mathrm{even}}}
\lambda_{\nummodes, \boldsymbol \mu}^{-1}
\lambda_{\nummodes, \boldsymbol \mu'}^{-1}
h_{\boldsymbol \mu}
h_{\boldsymbol \mu'}
\trace{
\Gamma_{\boldsymbol \mu} \hat{\rho}
}
\trace{
\Gamma_{\boldsymbol \mu'} \hat{\rho}
}.
\end{align}
Therefore, by linearity, the theorem follows from
\begin{align}
&
\expectationover{\mathbf p, \mathbf b}{
\trace{
\Gamma_{\boldsymbol \mu} \hat{\rho}
}
\trace{
\Gamma_{\boldsymbol \mu'} \hat{\rho}
}
}
\\
&=
\frac{1}{(2\nummodes)!}
\sum_{\substack{
\boldsymbol \tau \in \mathcal C_{2\nummodes} \\
\mathbf p \in \perm(2\nummodes) \\
\mathbf b \in {\{0, 1\}}^{\nummodes}
}}
g_{\boldsymbol \tau}
\trace{
\Gamma_{\boldsymbol \tau}
{U(\mathbf p)}^{\dagger}
\ket{\mathbf b} \bra{\mathbf b}
{U(\mathbf p)}
}
\trace{
\Gamma_{\boldsymbol \mu} 
{U(\mathbf p)}^{\dagger}
\ket{\mathbf b} \bra{\mathbf b}
{U(\mathbf p)}
}
\trace{
\Gamma_{\boldsymbol \mu'}
{U(\mathbf p)}^{\dagger}
\ket{\mathbf b} \bra{\mathbf b}
{U(\mathbf p)}
}
\label{eq:correlation-expansion}
\\
&=
\frac{1}{(2\nummodes)!}
\sum_{\substack{
\boldsymbol \tau \in \mathcal C_{2\nummodes} \\
\mathbf p \in \perm(2\nummodes) \\
\mathbf b \in {\{0, 1\}}^{\nummodes}
}}
g_{\boldsymbol \tau}
\braket{\mathbf b |
\Gamma_{\mathbf p(\boldsymbol \tau)}
| \mathbf b}
\braket{\mathbf b|
\Gamma_{\mathbf p(\boldsymbol \mu)}
| \mathbf b}
\braket{\mathbf b|
\Gamma_{\mathbf p(\boldsymbol \mu')}
| \mathbf b}
\\
&=
\frac{1}{(2\nummodes)!}
\sum_{\substack{
\boldsymbol \tau \in \mathcal C_{2\nummodes} \\
\mathbf p : 
\mathbf p(\boldsymbol \tau),
\mathbf p(\boldsymbol \mu),
\mathbf p(\boldsymbol \mu')
\in \mathcal D^*_{2\nummodes}
\\
\mathbf b \in {\{0, 1\}}^{\nummodes}
}}
g_{\boldsymbol \tau}
{(-1)}^{
    \mathbf b \cdot \left(
    \bin(\mathbf p(\boldsymbol \tau)) + 
    \bin(\mathbf p(\boldsymbol \mu)) + 
    \bin(\mathbf p(\boldsymbol \mu'))
    \right)
}
\sign\left(\mathbf p(\boldsymbol \tau)\right)
\sign\left(\mathbf p(\boldsymbol \mu)\right)
\sign\left(\mathbf p(\boldsymbol \mu')\right)
\\
&=
\frac{2^{\nummodes}}{(2\nummodes)!}
\sum_{\substack{
\boldsymbol \tau : \bin(\mathbf p(\boldsymbol \tau)) = \bin(\mathbf p(\boldsymbol \mu)) \oplus \bin(\mathbf p(\boldsymbol \mu')) \\
\mathbf p : 
\mathbf p(\boldsymbol \tau),
\mathbf p(\boldsymbol \mu),
\mathbf p(\boldsymbol \mu')
\in \mathcal D^*_{2\nummodes}
}}
g_{\boldsymbol \tau}
\sign\left(\mathbf p(\boldsymbol \tau)\right)
\sign\left(\mathbf p(\boldsymbol \mu)\right)
\sign\left(\mathbf p(\boldsymbol \mu')\right)
\\
&=
2^{\nummodes}
g_{\boldsymbol \mu \boldsymbol \mu'}
\sum_{\substack{
\mathbf p : 
\mathbf p(\boldsymbol \mu),
\mathbf p(\boldsymbol \mu')
\in \mathcal D^*_{2\nummodes}
}}
\frac{1}{(2\nummodes)!}
\label{eq:correlation-sign}
\\
&=
2^{\nummodes}
g_{\boldsymbol \mu \boldsymbol \mu'}
\Pr_{\mathbf p}
\left[
\mathbf p(\boldsymbol \mu), 
\mathbf p(\boldsymbol \mu) \in \mathcal D^*_{2\nummodes}
\right]
\\
&=
2^{\nummodes} g_{\boldsymbol \mu \boldsymbol \mu'}
\lambda_{\nummodes, \boldsymbol \mu', \boldsymbol \mu'}.
\end{align}
In \cref{eq:correlation-sign}, we used the fact that for fixed $\boldsymbol \mu, \boldsymbol \mu', \mathbf p$, there is exactly one $\boldsymbol \tau \in \mathcal D^*_{2\nummodes}$ such that 
$\bin(\mathbf p(\boldsymbol \tau)) = 
\bin(\mathbf p(\boldsymbol \mu)) + 
\bin(\mathbf p(\boldsymbol \mu')) 
$, and that
\begin{flalign}
\Gamma_{\boldsymbol \mu \boldsymbol \mu'}
&=
\Gamma_{\boldsymbol \mu}
\Gamma_{\boldsymbol \mu'}
\\
&=
U^{\dagger} U 
\Gamma_{\boldsymbol \mu}
U^{\dagger} U 
\Gamma_{\boldsymbol \mu'}
U^{\dagger} U 
=
U^{\dagger} 
\Gamma_{\mathbf p(\boldsymbol \mu)}
\Gamma_{\mathbf p(\boldsymbol \mu')}
U 
\\
&=
\sign(\mathbf p(\boldsymbol \mu))
\sign(\mathbf p(\boldsymbol \mu'))
U^{\dagger}
\Gamma_{\seq(\bin(\mathbf p(\boldsymbol \mu)))}
\Gamma_{\seq(\bin(\mathbf p(\boldsymbol \mu')))}
U
\\
&=
\sign(\mathbf p(\boldsymbol \mu))
\sign(\mathbf p(\boldsymbol \mu'))
U^{\dagger}
\Gamma_{\seq(\bin(\mathbf p(\boldsymbol \mu)) \oplus \bin(\mathbf p(\boldsymbol \mu')))}
U
=
\sign(\mathbf p(\boldsymbol \mu))
\sign(\mathbf p(\boldsymbol \mu'))
U^{\dagger}
\Gamma_{\seq(\bin(\mathbf p(\boldsymbol \tau)))}
U
\\
&=
\sign(\mathbf p(\boldsymbol \tau))
\sign(\mathbf p(\boldsymbol \mu))
\sign(\mathbf p(\boldsymbol \mu'))
U^{\dagger}
\Gamma_{\mathbf p(\boldsymbol \tau)}
U
=
\sign(\mathbf p(\boldsymbol \tau))
\sign(\mathbf p(\boldsymbol \mu))
\sign(\mathbf p(\boldsymbol \mu'))
\Gamma_{\boldsymbol \tau}.
\end{flalign}

\end{proof}
 \subsection{Computation of the estimator}\label{sec:computation}
In this section we show how to efficiently compute the estimators of projectors and $X$-type operators, as captured by \cref{lem:computation}, restated below.

\computationlem*

\begin{proof}
First, note that for any operator $A$ and any fermionic Gaussian unitary $V$, conjugation by $V$ commutes with the inverse channel:
\begin{align}
\mathcal M^{-1}\left(U A U^{\dagger}\right)
&=
U
\mathcal M^{-1}\left( A \right)
U^{\dagger}.
\end{align}
Let $V$ be the fermionic Gaussian unitary that prepares $\ket{\psi} = V \ket{\mathbf 0}$.
Then 
\begin{align}
\trace{\mathcal M^{-1}_{\mathcal U} \left(\ket{\psi} \bra{\psi}\right) \ket{\phi} \bra{\phi}}
&=
\trace{\mathcal M^{-1}_{\mathcal U} \left(V \ket{\mathbf 0} \bra{\mathbf 0} V^{\dagger} \right) \ket{\phi} \bra{\phi}}
\\
&=
\trace{\mathcal M^{-1}_{\mathcal U} \left(\ket{\mathbf 0} \bra{\mathbf 0} \right) V^{\dagger} \ket{\phi} \bra{\phi} V}.
\end{align}
The difficulty is that $\mathcal M^{-1}_{\mathcal U} \left(\ket{\mathbf 0} \bra{\mathbf 0} \right)$ is not a matchgate.
However, as captured by the following lemma, it can be written as a sum of $\nummodes + 1$ matchgate tensor networks.
The proof is deferred to the end of the 
section.

\begin{lemma}\label{lem:channel-transform}
\begin{align}
\mathcal M^{-1}
\left(
\ket{\mathbf 0} \bra{\mathbf 0}
\right)
&=
\sum_{j=0}^{\nummodes}
c_j
{
\left[
\sum_{b \in \{0, 1\}}
\omega^{-j b}
\ket{b} \bra{b}
\right]
}^{\otimes \nummodes},
\end{align}
where
\begin{align}
c_j &= 
\frac{1}{\nummodes + 1}
\sum_{i=0}^{\nummodes}
\omega^{i j} f_{\nummodes, i}, \\
\omega &= \exp\left(2 \pi i / (\nummodes + 1)\right),\\
f_{\nummodes, i}
    &= 
\frac{
2^{-\nummodes}
}{
\binom{\nummodes}{i}
}
\sum_{a=0}^{i}
{(-1)}^a
\sum_{k=a}^{\nummodes - i + a}
\binom{2\nummodes}{2k}
\binom{k}{a}
\binom{\nummodes-k}{i-a}.
\end{align}
\end{lemma}
Let 
\begin{align}
\mathcal M^{-1}_{\mathcal U} \left(\ket{\mathbf 0} \bra{\mathbf 0} \right) &= \sum_{j=0}^{\nummodes} M_j
\end{align}
as in \cref{lem:channel-transform}.
Then
\begin{align}
\trace{\mathcal M^{-1}_{\mathcal U} \left(\ket{\psi} \bra{\psi}\right) \ket{\phi} \bra{\phi}}
&= 
\sum_{j=0}^{\nummodes} \trace{M_j V^{\dagger} \ket{\phi} \bra{\phi} V}
\end{align}
can be computed in $\poly(\nummodes)$ time because $M_j$, $\ket{\phi}$, and $V$ are all matchgate.

Now, suppose that $\ket{\psi}$ is an $\numelec$-electron Slater determinant for even $n$.
Then there is a number-preserving fermionic Gaussian unitary $V$ such that $\ket{\psi} = V \ket{1}^{\otimes \numelec} \ket{0}^{\otimes (\nummodes - \numelec)}$.
Thus
\begin{align}
\trace{\mathcal M^{-1}_{\mathcal U} \left(\ket{\mathbf 0} \bra{\psi}\right) \ket{\phi} \bra{\phi}}
&=
\trace{\mathcal M^{-1}_{\mathcal U} \left(V \ket{\mathbf 0, \mathbf 0} \bra{\mathbf 1, \mathbf 0} V^{\dagger} \right) \ket{\phi} \bra{\phi}}
\\
&=
\trace{\mathcal M^{-1}_{\mathcal U} \left(\ket{\mathbf 0, \mathbf 0} \bra{\mathbf 1, \mathbf 0} \right) V^{\dagger} \ket{\phi} \bra{\phi} V}.
\end{align}
Analogous to \cref{lem:channel-transform}, we can write
\begin{align}
\mathcal M^{-1}_{\mathcal U}\left(
\ket{\mathbf 0, \mathbf 0} \bra{\mathbf 1, \mathbf 0}
\right)
&= \sum_{j=0}^{\nummodes - \numelec} M_j
\end{align}
as a sum of matchgates, so that 
\begin{align}
\trace{\mathcal M^{-1}_{\mathcal U} \left(\ket{\mathbf 0} \bra{\psi}\right) \ket{\phi} \bra{\phi}}
&=
\sum_{j=0}^{\nummodes - \numelec}
\trace{M_j  V^{\dagger} \ket{\phi} \bra{\phi} V}
\end{align}
can be computed in $\poly(\nummodes)$ time.
\end{proof}

\begin{proof}[Proof of \cref{lem:channel-transform}]
We begin by expanding $\ket{\mathbf 0} \bra{\mathbf 0}$ in the Majorana basis and applying the inverse channel:
\begin{align}
\mathcal M^{-1} \left(\ket{\mathbf 0} \bra{\mathbf 0}\right)
&=
\mathcal M^{-1} \left(
2^{-\nummodes}
\sum_{\mathbf z \in {\{0, 1\}}^{\nummodes}}
\Gamma_{\seq(\mathbf z)}
\right)
\\
&=
2^{-\nummodes}
\sum_{k=0}^{\nummodes}
\lambda_{\nummodes, k}^{-1}
\sum_{\mathbf z \in {\{0, 1\}}^{\nummodes}: \norm{\mathbf z}_1 = k}
\Gamma_{\seq(\mathbf z)}
\\
&=
2^{-\nummodes}
\sum_{k=0}^{\nummodes}
\lambda_{\nummodes, k}^{-1}
\sum_{\mathbf z \in {\{0, 1\}}^{\nummodes}: \norm{\mathbf z}_1 = k}
{(-1)}^{\mathbf b \cdot \mathbf z}
\sum_{\mathbf b \in {\{0, 1\}}^{\nummodes}}
\ket{\mathbf b} \bra{\mathbf b}
\\
&=
\sum_{\mathbf b \in {\{0, 1\}}^{\nummodes}}
\left[
2^{-\nummodes}
\sum_{k=0}^{\nummodes}
\lambda_{\nummodes, k}^{-1}
\sum_{a=\max\{0, k + \norm{\mathbf b}_1 - \nummodes\}}^{\min\{k, \norm{\mathbf b}_1\}}
{(-1)}^a
\binom{\norm{\mathbf b}_1}{a}
\binom{\nummodes - \norm{\mathbf b}_1}{k - a}
\right]
\ket{\mathbf b} \bra{\mathbf b}
\\
&=
\sum_{\mathbf b \in {\{0, 1\}}^{\nummodes}}
f_{\nummodes, \norm{\mathbf b}_1}
\ket{\mathbf b} \bra{\mathbf b}
,
\end{align}
where
\begin{align}
f_{\nummodes, i}
&=
2^{-\nummodes}
\sum_{k=0}^{\nummodes}
\lambda_{\nummodes, k}^{-1}
\sum_{a=\max\{0, k + i - \nummodes\}}^{\min\{k, i\}}
{(-1)}^a
\binom{i}{a}
\binom{\nummodes - i}{k - a}
\\
&=
2^{-\nummodes}
\sum_{a=0}^{i}
\sum_{k=a}^{\nummodes - i + a}
\lambda_{\nummodes, k}^{-1}
{(-1)}^a
\binom{i}{a}
\binom{\nummodes - i}{k - a}
\\
&=
\frac{
2^{-\nummodes}
}{
\binom{\nummodes}{i}
}
\sum_{a=0}^{i}
{(-1)}^a
\sum_{k=a}^{\nummodes - i + a}
\binom{2\nummodes}{2k}
\binom{k}{a}
\binom{\nummodes-k}{i-a}.
\end{align}
Then for all $i \in \{0, \ldots, \nummodes\}$,
\begin{align}
\sum_{j=0}^{\nummodes}
\omega^{-ij}
c_j 
&=
\frac{1}{\nummodes + 1}
\sum_{j, i'=0}^{\nummodes}
\omega^{-ij}
\omega^{j i'}
f_{i'}
\\
&=
\frac{1}{\nummodes + 1}
\sum_{j, i'=0}^{\nummodes}
f_{i'}
\omega^{j(i'-i)}
\\
&=
f_i.
\end{align}
Finally,
\begin{align}
\mathcal M^{-1}
\left(
\ket{\mathbf 0} \bra{\mathbf 0}
\right)
&=
\sum_{\mathbf b \in {\{0, 1\}}^{\nummodes}}
f_{\nummodes, \norm{\mathbf b}_1}
\ket{\mathbf b} \bra{\mathbf b}
\\
&=
\sum_{j=0}^{\nummodes}
c_j
\sum_{\mathbf b \in {\{0, 1\}}^{\nummodes}}
{\left(
\omega^{-j}
\right)}^{\norm{\mathbf b}_1}
\ket{\mathbf b} \bra{\mathbf b}
\\
&=
\sum_{j=0}^{\nummodes}
c_j
{
\left[
\sum_{b \in \{0, 1\}}
\omega^{-j b}
\ket{b} \bra{b}
\right]
}^{\otimes \nummodes}.
\end{align}
\end{proof}
 \subsection{Shadow norm of projector onto zero state}\label{sec:shadow-norm-bounds}
In this section we prove \cref{lem:projector-norm}, restated below.

\projectornormlem*

\begin{proof}[Proof of \cref{lem:projector-norm}]
We start by noting that 
\begin{align}
\ket{\mathbf 0}\bra{\mathbf 0}
&=
2^{-\nummodes}
\sum_{k = 0}^{\nummodes}
\sum_{\boldsymbol \mu \in \mathcal D_{2\nummodes, 2k}}
\Gamma_{\boldsymbol \mu}.
\label{eq:zero-projector-expansion}
\end{align}
Thus
\begin{flalign}
&
\norm{
\ket{\mathbf 0}\bra{\mathbf 0}
}_{\mathcal U}
\\
&
\leq
2^{-2\nummodes}
\sum_{
\boldsymbol \mu, \boldsymbol \mu' \in \mathcal D_{2\nummodes}
}
\kappa_{\nummodes, \boldsymbol \mu, \boldsymbol \mu'}
&\text{\cref{eq:zero-projector-expansion}, \cref{cor:shadow-norm-bound}}
\\
&=
2^{-2\nummodes}
\sum_{k= 0}^{\nummodes}
\binom{\nummodes}{k}
\sum_{a=0}^k
\binom{k}{a}
\sum_{b=0}^{\nummodes - k}
\binom{\nummodes - k}{b}
\kappa_{\nummodes, k, a + b, a}
\\
&\leq
2^{-\nummodes}
\sum_{k= 0}^{\nummodes}
\binom{\nummodes}{k}
\sum_{a=0}^k
\binom{k}{a}
\sum_{b=0}^{\nummodes - k}
\binom{\nummodes - k}{b}
\frac{
\binom{k}{a}
\binom{\nummodes - k}{b}
}{
\binom{2k}{2a}
\binom{2(\nummodes-k)}{2b}
}
&\text{\cref{fact:kappa-bound}}
\label{eq:projector-variance-expand-multinom}
\\
&=
2^{-\nummodes}
\sum_{k= 0}^{\nummodes}
\binom{\nummodes}{k}
\sum_{a=0}^k
\left(
\frac{
\binom{k}{a}^2
}{
\binom{2k}{2a}
}
\right)
\sum_{b=0}^{\nummodes - k}
\left(
\frac{
\binom{\nummodes - k}{b}^2
}{
\binom{2(\nummodes - k)}{2b}
}
\right)
\label{eq:projector-variance-superbinom-bound}
\\
&=
2^{-\nummodes}
\sum_{k= 0}^{\nummodes}
\binom{\nummodes}{k}
\sum_{a=0}^k
\left(
\frac{
\binom{2a}{a}
\binom{2(k-a)}{k-a}
}{
\binom{2k}{k}
}
\right)
\sum_{b=0}^{\nummodes - k}
\left(
\frac{
\binom{2b}{b}
\binom{2(\nummodes-k-b)}{\nummodes - k - b}
}{
\binom{2(\nummodes - k)}{\nummodes - k}
}
\right)
\\
&=
2^{-\nummodes}
\sum_{k= 0}^{\nummodes}
\frac{
\binom{\nummodes}{k}
}{
\binom{2k}{k}
\binom{2(\nummodes - k)}{\nummodes - k}
}
\left[
\sum_{a=0}^k
\binom{2a}{a}
\binom{2(k-a)}{k-a}
\right]
\left[
\sum_{b=0}^{\nummodes - k}
\binom{2b}{b}
\binom{2(\nummodes-k-b)}{\nummodes - k - b}
\right]
\\
&=
2^{-\nummodes}
\sum_{k= 0}^{\nummodes}
\frac{
\binom{\nummodes}{k}
}{
\binom{2k}{k}
\binom{2(\nummodes - k)}{\nummodes - k}
}
2^{2k}
2^{2(\nummodes-k)}
=
2^{\nummodes}
\sum_{k= 0}^{\nummodes}
\frac{
\binom{\nummodes}{k}
}{
\binom{2k}{k}
\binom{2(\nummodes - k)}{\nummodes - k}
}
\\
&\leq
2^{\nummodes}
\sum_{k= 0}^{\nummodes}
\binom{\nummodes}{k}
\sqrt{k}
2^{1-2k}
\sqrt{\nummodes - k}
2^{1-2(\nummodes-k)}
\\
&=
2^{2-\nummodes}
\sum_{k= 0}^{\nummodes}
\binom{\nummodes}{k}\sqrt{k(\nummodes - k)}
\\
&=
\sum_{k= 0}^{\nummodes}
2^{1-\nummodes}\nummodes
\binom{\nummodes}{k}
\left(
2
\sqrt{\frac{k}{m} \left(1 - \frac{k}{m}\right)}
\right)
\\
&\leq
\sum_{k= 0}^{\nummodes}
2^{1-\nummodes}\nummodes
\binom{\nummodes}{k}
=
2 \nummodes.
\end{flalign}
\end{proof}

 \subsection{Numerical evaluation of the shadow norm}\label{sec:numerics}
\begin{figure}[t!]
\centering
\subcaptionbox{$F_0(m)$ and $f(m, 0)$}{
\includegraphics[width=0.45\textwidth]{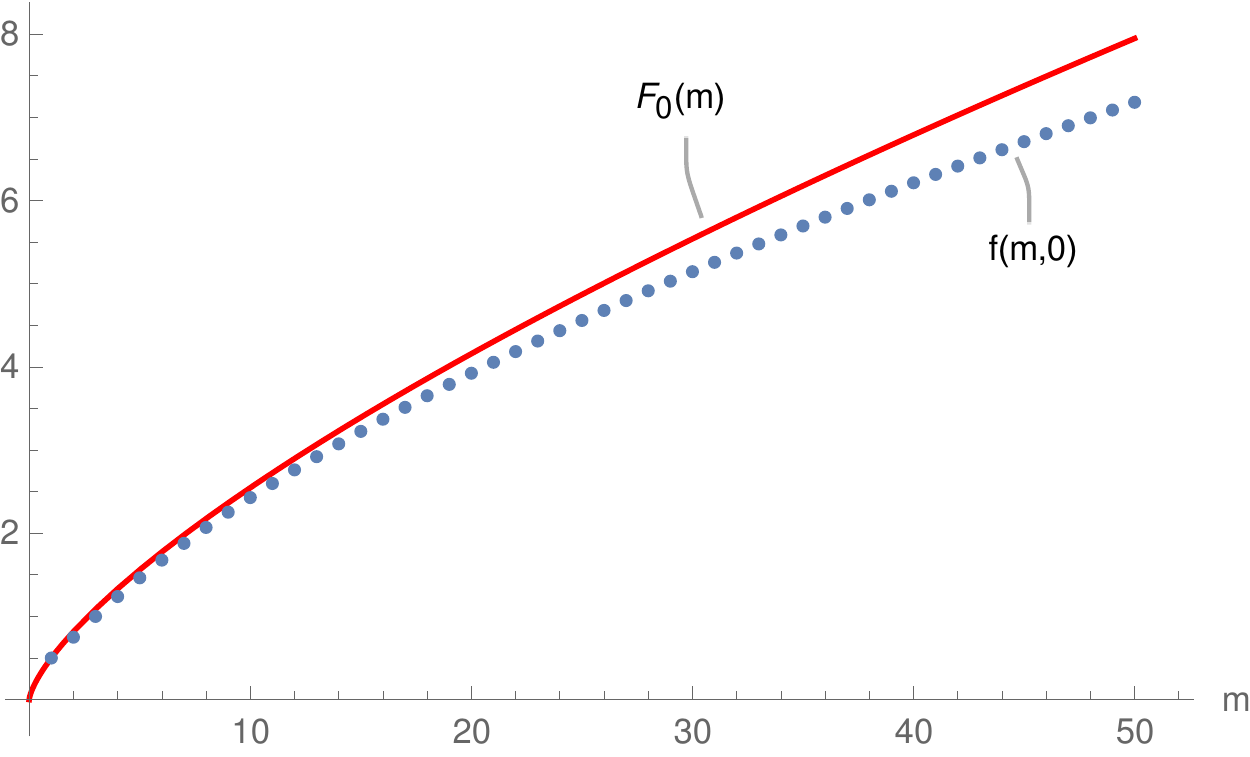}
}
\label{fig:F0}
\subcaptionbox{$F_1(m)$ and $f(m, m)$}{
\includegraphics[width=0.45\textwidth]{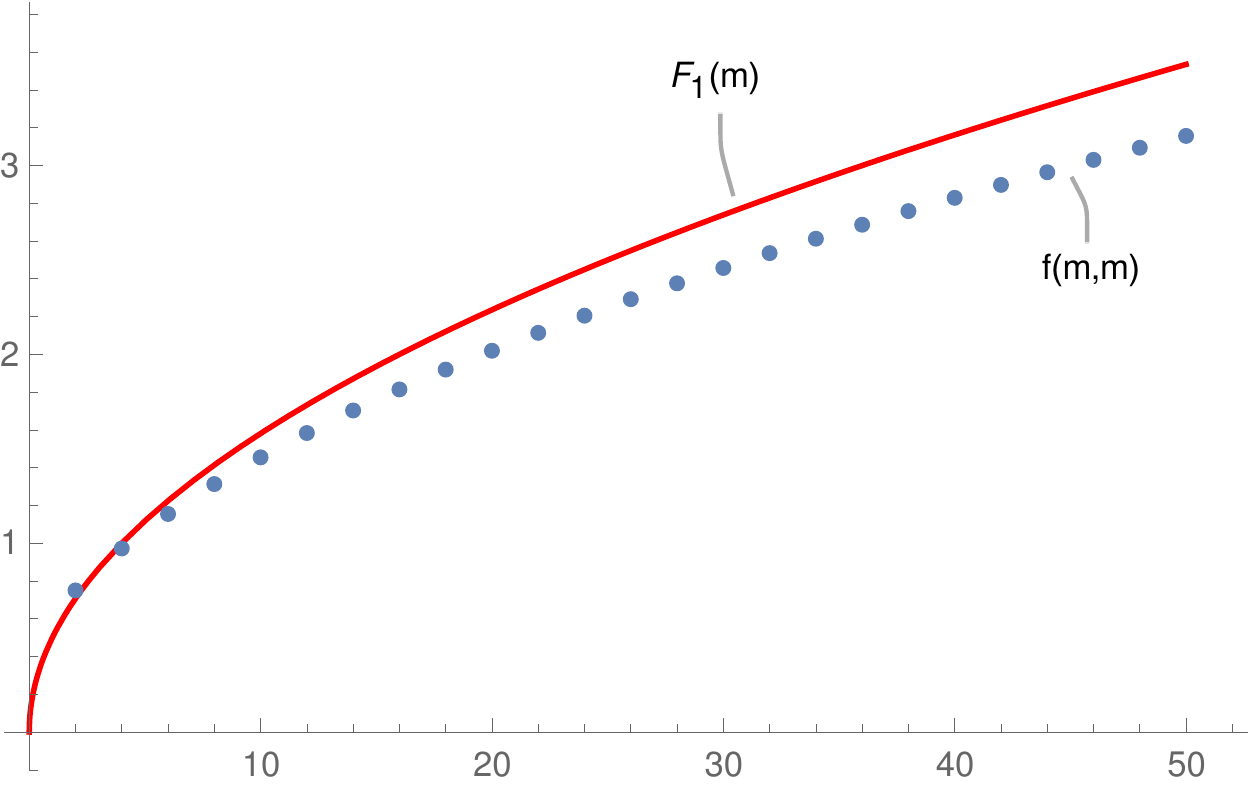}
}
\label{fig:F1}
\caption{Numerical evaluation of the function $f(\nummodes, \numelec)$ at $\numelec  = 0$ and $\numelec=\nummodes$, together with conjectured upper bounds $F_0(\nummodes)$ and $F_1(\nummodes)$, respectively.}\label{fig:projector-norm}
\end{figure}
 
In this section we provide numerical evidence that
\begin{align}
\norm{
{\left(
\ket{0} \bra{1}
\right)}^{\numelec}
{\left(
\ket{1} \bra{1}
\right)}^{\nummodes - \numelec}
}_{\mathcal U}
= O (\sqrt{\nummodes}).
\end{align}
We start by noting that 
\begin{align}
{\left(
\ket{0} \bra{1}
\right)}^{\numelec}
{\left(
\ket{1} \bra{1}
\right)}^{\nummodes - \numelec}
&= 
2^{-\nummodes}
{(-i)}^{\binom{\numelec}{2}}
\sum_{\substack{
\mathbf x \in {\{0, 1\}}^{\numelec} \\
\mathbf z \in {\{0, 1\}}^{\nummodes - \numelec} 
}}
i^{\norm{\mathbf x}_1}
\Gamma_{\seqx(\mathbf x)}
\otimes \Gamma_{\seq(\mathbf z)}
\\
&=
\sum_{\substack{
\mathbf x \in {\{0, 1\}}^{\numelec} \\
\mathbf z \in {\{0, 1\}}^{\nummodes - \numelec} 
}}
q_{\mathbf x, \mathbf z}
\Gamma_{\seqxz(\mathbf x, \mathbf z)},
\end{align}
where we define
\begin{align}
q_{\mathbf x, \mathbf z}
&=
2^{-\nummodes}
{(-i)}^{\binom{\numelec}{2}}
i^{\norm{\mathbf x}_1}
,
&
\seqxz(\mathbf x, \mathbf z)
&=
\Gamma_{\seqx(\mathbf x)}
\otimes \Gamma_{\seq(\mathbf z)}.
\end{align}
Using \cref{cor:shadow-norm-bound},
\begin{align}
&\norm{
{\left(
\ket{0} \bra{1}
\right)}^{\numelec}
{\left(
\ket{1} \bra{1}
\right)}^{\nummodes - \numelec}
}_{\mathcal U}^2
\\
&\leq 
\sum_{\substack{
\mathbf x, \mathbf x' \in {\{0, 1\}}^{\numelec} : \text{$\mathbf x \cdot \mathbf x'$ even}\\
\mathbf z, \mathbf z' \in {\{0, 1\}}^{\nummodes - \numelec} 
}}
\abs{
q_{\mathbf x, \mathbf z}
}
\abs{
q_{\mathbf x', \mathbf z'}
}
\kappa^{-1}_{\nummodes, \seqxz(\mathbf x, \mathbf z), \seqxz(\mathbf x', \mathbf z')}
\\
&=
2^{-2\nummodes}
\sum_{\substack{
\mathbf x, \mathbf x' \in {\{0, 1\}}^{\numelec} : \text{$\mathbf x \cdot \mathbf x'$ even}\\
\mathbf z, \mathbf z' \in {\{0, 1\}}^{\nummodes - \numelec} 
}}
\kappa^{-1}_{\nummodes, \seqxz(\mathbf x, \mathbf z), \seqxz(\mathbf x', \mathbf z')}
\\
&=
2^{\numelec-2\nummodes -1}
\sum_{a_1 = 0}^{\numelec / 2}
\binom{\numelec}{2a_1}
\sum_{k_2 = 0}^{\nummodes - \numelec}
\binom{\nummodes - \numelec}{k_2}
\sum_{a_2=0}^{k_2}
\binom{k_2}{a_2}
\sum_{b_2=0}^{\nummodes - \numelec - k_2}
\binom{\nummodes - \numelec - k_2}{b_2}
\kappa_{\nummodes, \frac{\numelec}{2} + k_2, \frac{\numelec}{2} + a_2 + b_2, a_1 + a_2}
\\
&=f(\nummodes, \numelec).
\end{align}
The bound $f(\nummodes, \numelec)$ is plotted in \cref{fig:projector-norm} for $\numelec \in \{0, \nummodes\}$,
together with the conjectured upper bounds
\begin{align}
f(\nummodes, \numelec) 
\overset{?}{\leq}
f(\nummodes, 0) 
&
\overset{?}{\leq}
F_0(\nummodes) = \frac12 \nummodes^{1/\sqrt{2}},\\
f(\nummodes, \nummodes)
&
\overset{?}{\leq}
F_1(\nummodes) = \frac12 \nummodes^{1/2}.
\end{align}
For $\nummodes \leq 50$, $f(\nummodes, \numelec)$ is monotonically decreasing with $\numelec$.

  \section{Smaller ensembles yielding the same channel}\label{sec:smaller-ensembles}

\begin{proof}[Proof of~\cref{thm:smaller-ensembles}]
Let $\mathcal F \in \perm(2\nummodes)$ be the set of $2^{\nummodes}$ permutations that swap elements only within each pair.
For each $\mathbf f \in \mathcal F$, $U(\mathbf f)$ is a product of single-qubit X gates.
For any permutation 
$\mathbf p \in \perm(2\nummodes)$ and 
$\mathbf f \in \mathcal F$ such that $U(\mathbf f) = X^{\otimes \mathbf x}$ 
\begin{align}
&
\sum_{\mathbf b \in {\{0, 1\}}^{\nummodes}}
\braket{ \mathbf b | U(\mathbf f \cdot \mathbf p)  \rho {U(\mathbf f \cdot \mathbf p)}^{\dagger}| \mathbf b}
{U(\mathbf f \cdot \mathbf p)}^{\dagger}    \ket{\mathbf b}\bra{\mathbf b} U(\mathbf f \cdot \mathbf p)
\\
&=
\sum_{\mathbf b \in {\{0, 1\}}^{\nummodes}}
\braket{ \mathbf b | 
U(\mathbf f) U(\mathbf p)  \rho 
{U(\mathbf p)}^{\dagger} {U(\mathbf f)}^{\dagger}| \mathbf b}
{U(\mathbf p)}^{\dagger} {U(\mathbf f)}^{\dagger}    
\ket{\mathbf b}\bra{\mathbf b} U(\mathbf f) U(\mathbf p)
\\
&=
\sum_{\mathbf b \in {\{0, 1\}}^{\nummodes}}
\braket{ \mathbf b \oplus \mathbf x| 
U(\mathbf p)  \rho 
{U(\mathbf p)}^{\dagger}| \mathbf b \oplus \mathbf x}
{U(\mathbf p)}^{\dagger}    
\ket{\mathbf b \oplus \mathbf x}\bra{\mathbf b \oplus \mathbf x} U(\mathbf p)
\\
&=
\sum_{\mathbf b \in {\{0, 1\}}^{\nummodes}}
\braket{ \mathbf b | U(\mathbf p)  \rho {U(\mathbf p)}^{\dagger}| \mathbf b}
{U(\mathbf p)}^{\dagger}    \ket{\mathbf b}\bra{\mathbf b} U(\mathbf p).
\end{align}

Let $\mathcal S \in \perm(2\nummodes)$ be the set of $\nummodes!$ permutations that swap adjacent pairs together.
For each $\mathbf s \in \mathcal S$, $U(\mathbf s)$ is  generated by standard SWAP gates.
In a slight abuse of notation, for a bitstring $\mathbf x \in {\{0, 1\}}^{\nummodes}$, let $\mathbf s(\mathbf x) \in {\{0, 1\}}^{\nummodes}$ be the corresponding swapped bitstring.
For any $\mathbf p \in \perm(2\nummodes)$ and 
$\mathbf s \in \mathcal S$,
\begin{align}
&
\sum_{\mathbf b \in {\{0, 1\}}^{\nummodes}}
\braket{ \mathbf b | U(\mathbf s \cdot \mathbf p)  \rho {U(\mathbf s \cdot \mathbf p)}^{\dagger}| \mathbf b}
{U(\mathbf s \cdot \mathbf p)}^{\dagger}    \ket{\mathbf b}\bra{\mathbf b} U(\mathbf s \cdot \mathbf p)
\\
&=
\sum_{\mathbf b \in {\{0, 1\}}^{\nummodes}}
\braket{ \mathbf b | 
U(\mathbf s) U(\mathbf p)  \rho 
{U(\mathbf p)}^{\dagger} {U(\mathbf s)}^{\dagger}| \mathbf b}
{U(\mathbf p)}^{\dagger} {U(\mathbf s)}^{\dagger}    
\ket{\mathbf b}\bra{\mathbf b} U(\mathbf s) U(\mathbf p)
\\
&=
\sum_{\mathbf b \in {\{0, 1\}}^{\nummodes}}
\braket{ \mathbf s^{-1} (\mathbf b) | 
U(\mathbf p)  \rho 
{U(\mathbf p)}^{\dagger}| \mathbf s^{-1}(\mathbf b)}
{U(\mathbf p)}^{\dagger}    
    \ket{\mathbf s^{-1}(\mathbf b)}\bra{\mathbf s^{-1}(\mathbf b)} U(\mathbf p)
\\
&=
\sum_{\mathbf b \in {\{0, 1\}}^{\nummodes}}
\braket{ \mathbf b | U(\mathbf p)  \rho {U(\mathbf p)}^{\dagger}| \mathbf b}
{U(\mathbf p)}^{\dagger}    \ket{\mathbf b}\bra{\mathbf b} U(\mathbf p).
\end{align}

Every permutation can be uniquely decomposed:
\begin{align}
\perm(2\nummodes) &=
\mathcal F \cdot \mathcal S \cdot \mathcal P
=
\left\{
\mathbf f \cdot \mathbf s \cdot \mathbf p
\middle| (\mathbf f, \mathbf s, \mathbf p) \in \mathcal F \times \mathcal S \times \mathcal P
\right\}.
\end{align}

Finally,
\begin{align}
\mathcal M_{\mathcal P}(\rho) &= 
\expectationover*{
\mathbf p \in \mathcal P
}{
\sum_{\mathbf b \in {\{0, 1\}}^{\nummodes}}
\braket{ \mathbf b | U(\mathbf p)  \rho {U(\mathbf p)}^{\dagger}| \mathbf b}
{U(\mathbf p)}^{\dagger}    \ket{\mathbf b}\bra{\mathbf b} U(\mathbf p)
}
\\
&=
\expectationover*{
\substack{
\mathbf f \in \mathcal F \\
\mathbf s \in \mathcal S \\
\mathbf p \in \mathcal P
}
}{
\sum_{\mathbf b \in {\{0, 1\}}^{\nummodes}}
\braket{ \mathbf b | 
U(\mathbf f \cdot \mathbf s \cdot \mathbf p)  \rho 
{U(\mathbf f \cdot \mathbf s \cdot \mathbf p)}^{\dagger}| \mathbf b}
{U(\mathbf f \cdot \mathbf s \cdot \mathbf p)}^{\dagger}    
\ket{\mathbf b}\bra{\mathbf b} U(\mathbf f \cdot \mathbf s \cdot \mathbf p)
}
\\
&=
\expectationover*{
\mathbf p \in \perm(2\nummodes)
}{
\sum_{\mathbf b \in {\{0, 1\}}^{\nummodes}}
\braket{ \mathbf b | U(\mathbf p)  \rho {U(\mathbf p)}^{\dagger}| \mathbf b}
{U(\mathbf p)}^{\dagger}    \ket{\mathbf b}\bra{\mathbf b} U(\mathbf p)
}
= \mathcal M(\rho).
\end{align}
\end{proof}
 \section{Learning a Slater determinant}\label{sec:slater-determinant}

In this section, we prove \cref{thm:slater-determinant-learning}, restated below.
The essence of the proof is that a Slater determinant is uniquely specified by its 1-RDMs.
However, given just copies of a Slater determinant, a learner can only approximately learn the 1-RDMs.
The technical work then is simply to give a procedure to extract a Slater determinant from the approximated 1-RDMs and to quantify how the approximation error in the estimated 1-RDMs affects the fidelity of the learned state with the target state.

\sdlearningthm*

\begin{proof}[Proof of \cref{thm:slater-determinant-learning}]
Suppose we are given samples of the Slater determinant $\ket{\psi}$.
By definition, there is some unitary $U$ such that 
$\ket{\psi}= b^{\dagger}_1 \cdots b_{\numelec}^{\dagger} \ket{\mathbf 0}$, where $b_i = \sum_{j=1}^{\nummodes} U_{i, j} a_j$.
Note that this unitary is not unique, and that
\begin{align}
a_i b_1^{\dagger} &= U_{1,i}^* - b_1^{\dagger} a_i, \\
a_i b_1^{\dagger} b_2^{\dagger}&= 
\left(U_{1,i}^* - b_1^{\dagger} a_i \right) b_2^{\dagger}
=
U_{1, i}^* b^{\dagger}_2 - U_{2, i}^* b_1^{\dagger} + b_1^{\dagger} b_2^{\dagger} a_i
,
\\
\vdots
\quad
&\hspace{10em}\vdots
\nonumber
\\
a_i b_1^{\dagger} \cdots b_{\numelec}^{\dagger} &=
-\sum_{j=1}^{\numelec} U_{j, i}^* {(-1)}^j b_1^{\dagger} \cdots b_{j-1}^{\dagger}  b_{j+1}^{\dagger} \cdots b_{\numelec}^{\dagger} 
+ {(-1)}^{\numelec} b_1^{\dagger} \cdots b_{\numelec}^{\dagger} a_i
.
\end{align}
Let
\begin{align}
\Pi_{\nummodes, \numelec} &= \sum_{l=1}^{\numelec} \ket{l}\bra{l}
\end{align}
be the $\nummodes \times \nummodes$ projector onto the first $\numelec$ entries and
let $R$ be the $\nummodes \times \nummodes$ (Hermitian) matrix of expectation values of the operators $a_i^{\dagger} a_j$ with entries
\begin{align}
&R_{j, i} 
= \Tr\left[a^{\dagger}_i a_j \ket{\psi}\bra{\psi}\right]
= \Tr\left[a^{\dagger}_i a_j b_1^{\dagger} \cdots b_{\numelec}^{\dagger} \ket{\mathbf 0} \bra{\mathbf 0} b_{\numelec} \cdots b_1\right]
= 
\braket{\mathbf 0| 
b_{\numelec} \cdots b_1
a^{\dagger}_i
a_j
b_1^{\dagger} \cdots b_{\numelec}^{\dagger} 
| \mathbf 0}
\\
&=
\bra{\mathbf 0} 
\left(
-\sum_{l=1}^{\numelec} U_{l, i} {(-1)}^l b_{\numelec} \cdots b_{l+1} b_{l-1} \cdots b_1
+ {(-1)}^{\numelec} a_i^{\dagger} b_{\numelec} \cdots b_1
\right)
\\
&\qquad\qquad
\left(
-\sum_{l'=1}^{\numelec} U_{l', j}^* {(-1)}^{l'} b_1^{\dagger} \cdots b_{l'-1}^{\dagger}  b_{l'+1}^{\dagger} \cdots b_{\numelec}^{\dagger} 
+ {(-1)}^{\numelec} b_1^{\dagger} \cdots b_{\numelec}^{\dagger} a_j
\right)
\ket{\mathbf 0}
\\
&=
\sum_{l, l'=1}^{\numelec} 
U_{l, i} 
U_{l', j}^*
{(-1)}^{l+l'} 
\braket{\mathbf 0| 
b_{\numelec} \cdots b_{l+1} b_{l-1} \cdots b_1
b_1^{\dagger} \cdots b_{l'-1}^{\dagger}  b_{l'+1}^{\dagger} \cdots b_{\numelec}^{\dagger} 
| \mathbf 0}
\\
&=
\sum_{l=1}^{\numelec} 
U_{l, i}
U_{l, j}^*
=
\sum_{l=1}^{\numelec} 
{\left(
U^{\dagger}
\right)}_{j, l}
U_{l, i}
=
{\left(
U^{\dagger}
\Pi_{\nummodes, \numelec}
U
\right)}_{j, i},
\end{align}
i.e.,
\begin{align}
R &= U^{\dagger} \Pi_{n,k } U.
\label{eq:R-def}
\end{align}
The matrix $R$ has all the information we need to uniquely specify $\ket{\psi}$, but we can't learn it exactly.
Noting that
\begin{align}
a^{\dagger}_i a_j
&= \frac14 \left[
i \Gamma_{(2i-1, 2j-1)} + i \Gamma_{(2i, 2j)} + \Gamma_{(2i, 2j-1)} - \Gamma_{(2i-1, 2j)}
\right],
\end{align}
we see that an additive $\epsilon_{\mathrm{shdw}}$-approximation to all degree-$2$ Majorana operators (i.e., the 1-RDMs) leads to an
additive $\epsilon_{\mathrm{shdw}}$-approximation to all entries of $R$.
Therefore, for any $\epsilon_{\mathrm{shdw}} > 0$ (to be set later) and with probability at least $1 - \delta$,
we can get an estimate $\estimate{R}$ such that 
$\left|\estimate{R}_{i, j} - R_{i, j}\right| \leq \epsilon_{\mathrm{shdw}}$ for all $1 \leq i, j \leq \nummodes$
using $O(\nummodes \log (\nummodes  / \delta) \epsilon_{\mathrm{shdw}}^{-2})$ samples of $\ket{\psi}$ by estimating the 1-RDMs according to \cref{thm:k-RDMs}.
Note that $\estimate{R}$ is Hermitian by construction.

The remainder of the proof shows how to use our estimate $\estimate{R}$ to deduce a classical description of a Slater determinant $\ket{\estimate{\psi}}$ such that ${\left|\braket{\psi | \estimate{\psi}}\right|}^2 \geq 1 - \epsilon_{\mathrm{fid}}$.
First, we find $\estimate{U}$ such that 
\begin{align}
\estimate{U} \estimate{R} \estimate{U}^{\dagger} &= \Lambda = \diag(\lambda_1 \cdots, \lambda_n),
\label{eq:U-est-def}
\end{align}
where $\lambda_1 \geq \lambda_2 \geq \cdots \geq \lambda_n \in \mathbb R$.\footnote{We will assume all classical arithmetic and linear algebra can be done with arbitrary precision.}
Because $\estimate{R}$ is Hermitian, there is always such a unitary $\estimate{U}$.
The rows of $\estimate{U}$ are the (conjugate transpose of the) eigenvectors of $\estimate{R}$, in decreasing order of the corresponding eigenvalues.
Let $\estimate{b}_i = \sum_{j=1}^{\nummodes} \estimate{U}_{i, j} a_j$.
Our estimate of the state will be
\begin{align}
\ket{\estimate{\psi}} &= \estimate{b}_1 \cdots \estimate{b}_{\numelec} \ket{\mathbf 0}.
\end{align}

To bound the fidelity between $\ket{\psi}$ and $\ket{\estimate{\psi}}$, we start by defining 
the errors
\begin{align}
E &= \estimate{R} - R,
&
D &= \Lambda - \Pi_{\nummodes, \numelec} = \diag(d_1, \ldots, d_{\nummodes}).
\end{align}
Note that $E$ is Hermitian, with entries at most $\epsilon_{\mathrm{shdw}}$ in magnitude, and $d_i \in \mathbb R$ for all $i \in [\nummodes]$.
With these, we can rewrite \cref{eq:U-est-def} as
\begin{align}
\estimate{U} R \estimate{U}^{\dagger} &=  \Pi_{\nummodes, \numelec} + D - \estimate{E},
&
&
\text{
where
}
&
\estimate{E} &= \estimate{U} E \estimate{U}^{\dagger}.
\end{align}
Plugging in \cref{eq:R-def}, we get
\begin{align}
\estimate{U} U^{\dagger} \Pi_{\nummodes, \numelec} U \estimate{U}^{\dagger}
&=  \Pi_{\nummodes, \numelec} + D - \estimate{E}.
\end{align}
Taking the first $\numelec$ rows and columns yields
\begin{align}
{\left(\estimate{U} U^{\dagger}\right)}_{[\numelec], [\numelec]}
{\left(\estimate{U} U^{\dagger}\right)}_{[\numelec], [\numelec]}^{\dagger}
&=  I + {\left(D - \estimate{E}\right)}_{[\numelec], [\numelec]}.
\label{eq:UU-IDE}
\end{align}
To finish, we will use the following lemma, whose proof is deferred to the end of the section.
It translates the approximation error of $\estimate{R}$ into a bound on $D - \estimate{E}$.
\begin{lemma}\label{lem:DE-norm}
For $R$ and $\estimate{R}$, let $U$, $\estimate{U}$, $D$, $\estimate{E}$ be defined as above.
If $\max_{i, j} \abs{R_{i, j} - \estimate{R}_{i, j}} \leq \epsilon_{\mathrm{shdw}} < 1 / (2 \nummodes^3)$, then
\begin{align}
\norm{D - \estimate{E}}
&\leq 
2 \nummodes^3 \epsilon_{\mathrm{shdw}}.
\end{align}
\end{lemma}

Finishing up the proof of \cref{thm:slater-determinant-learning}, we get that the overlap of our estimate $\ket{\estimate{\psi}}$ with the target state $\ket{\psi}$ is 
\begin{flalign}
{\left|\braket{\estimate{\psi} | \psi}\right|}^2
&=
{\left|
\det\left({\left(\estimate{U} U^{\dagger}\right)}_{[\numelec], [\numelec]}\right)
\right|}^2
\\
&=
\det\left(
{\left(\estimate{U} U^{\dagger}\right)}_{[\numelec], [\numelec]}
{\left(\estimate{U} U^{\dagger}\right)}_{[\numelec], [\numelec]}^{\dagger}
\right)
\label{eq:fidelity}
\\
&= \det \left(I + {\left(D - \estimate{E}\right)}_{[\numelec], [\numelec]}\right)
&\text{by \cref{eq:UU-IDE}}
\\
&\geq 
\norm{I + {\left(D - \estimate{E}\right)}_{[\numelec], [\numelec]}}^{\numelec}
\\
&\geq
{\left(1 - \left\|{\left(D - \estimate{E}\right)}_{[\numelec], [\numelec]}\right\|\right)}^{\numelec}
&\forall M, 1 = \norm{(I + M) - M} \leq \norm{I + M} + \norm{M}
\label{eq:fidelity-det-bound}
\\
&\geq
{\left(1 - \left\|D - \estimate{E}\right\|\right)}^{\numelec}
&
\forall M,
\norm{\Pi_{\nummodes, \numelec} M \Pi_{\nummodes, \numelec}} 
\leq \norm{M}
\label{eq:fidelity-submultiplicativity}
\\
&
\geq 1 - \numelec \left\|D - \estimate{E}\right\|
&\text{by Bernoulli's inequality}
\label{eq:fidelity-bernoulli}
\\
&
\geq  1 - 2 \numelec \nummodes^3 \epsilon_{\mathrm{shdw}}.
&\text{by \cref{lem:DE-norm}}
\label{eq:fidelity-lem-DE-norm}
\end{flalign}
Note that \cref{eq:fidelity} depends only on the first $\numelec$ rows of $U$ and $\estimate{U}$ but not on their ordering, as expected.
Setting
\begin{align}
\epsilon_{\mathrm{shdw}} &= \epsilon_{\mathrm{fid}} / \left(3 \numelec \nummodes^3\right)
\end{align}
ensures that ${\left|\braket{\estimate{\psi} | \psi}\right|}^2 \geq 1 - \epsilon_{\mathrm{fid}}$
and also that 
$2 \nummodes^3 \epsilon_{\mathrm{shdw}} = \epsilon_{\mathrm{fid}}  / \numelec \leq 1/2$, satisfying the precondition of \cref{lem:DE-norm}.
The number of samples is 
$O(\nummodes \log(\nummodes / \delta) / \epsilon_{\mathrm{shdw}}^2) = 
O(\numelec^2 \nummodes^7 \log(\nummodes / \delta) / \epsilon_{\mathrm{fid}}^2)$.
\end{proof}

The proof of \cref{lem:DE-norm} will make use of the Gershgorin circle theorem.
If our estimate $\estimate{R}$ of $R$ were exact, then the eigenvalues of $\estimate{R}$ would be $0$ and $1$, as they are for $R$.
In a sense, the eigenvalue $1$ subspace is exactly what we want to learn.
When $\estimate{R}$ is only close to $R$, then each eigenvalue is close to $0$ or $1$.
The Gershgorin circle theorem, stated below, allows us to bound how much error in $\estimate{R}$ we can tolerate before the two subspaces bleed too much into each other.
\begin{theorem}[Gershgorin circle theorem~\cite{gershgorin1931ueber}]\label{thm:gershgorin-circle}
Let $A = {\left(a_{i, j}\right)}_{1 \leq i, j \leq n}$ be an $n \times n$ matrix.
For $i = 1, \ldots, n$, let $K_i$ be the circle with center $a_{i, i}$ and radius $\sum_{k\neq i} \left| a_{i, k}\right|$.
Then all eigenvalues of $A$ are contained in $\bigcup_{i=1}^n K_i$.
Furthermore, for $I \subset [n]$, if $\bigcup_{i \in I} K_i$ is disjoint from $\bigcup_{i \notin I} K_i$, then the former contains exactly $|I|$ eigenvalues.
\end{theorem}

\begin{proof}[Proof of \cref{lem:DE-norm}]

The diagonal elements of $\Lambda = \diag(\lambda_1, \ldots, \lambda_{\nummodes}) = \estimate{U} \estimate{R} \estimate{U}^{\dagger}$ are the eigenvalues of $\estimate{R}$, which are the eigenvalues of 
\begin{align}
    F &= 
    U \estimate{R} U^{\dagger}
    =
    U R U^{\dagger}
    +
    U E U^{\dagger}
    =
    \Pi_{\nummodes, \numelec} + U E U^{\dagger}
    .
\end{align}
For $i \in [\nummodes]$, let 
\begin{align}
r_i = \sum_{j\neq i} \left|{\left(U E U^{\dagger}\right)}_{i, j}\right|
\end{align}
be the radii around $F_{i, i}$ defining the circle $K_i$ of the Gershgorin circle theorem (\cref{thm:gershgorin-circle}).
Define
\begin{align}
\epsilon_{\mathrm{EV}}
=
\nummodes^3 \epsilon_{\mathrm{shdw}}
\geq 
\nummodes
\max_{j, k}
\left\{
\sum_{l, l'=1}^{\nummodes}
{\left|U_{j, l} E_{l, l'}  U_{k, l'}^*\right|}
\right\}
\geq
\nummodes
\max_{j, k}
\left|{\left(U E U^{\dagger}\right)}_{j, k}\right|
\geq
\max_{j}
\left\{
\sum_{k=1}^{\nummodes}
\left|{\left(U E U^{\dagger}\right)}_{j, k}\right|
\right\}
.
\end{align}
Then the individual Gershgorin discs are 
\begin{align}
&F_{i, i} - r_{i, i}  & & F_{i, i} + r_{i, i} \\
&= 
{\left(\Pi_{\nummodes, \numelec}\right)}_{i, i}
+ {\left(U E U^{\dagger}\right)}_{i, i}
-
\sum_{j \neq i} 
\left| {\left(U E U^{\dagger}\right)}_{i, j}\right|
&&= {\left(\Pi_{\nummodes, \numelec}\right)}_{i, i}
+ {\left(U E U^{\dagger}\right)}_{i, i}
+ \sum_{j \neq i} 
\left| {\left(U E U^{\dagger}\right)}_{i, j}\right|
\\
&\geq 
{\left(\Pi_{\nummodes, \numelec}\right)}_{i, i}
- \sum_{j =1}^{\nummodes}
\left| {\left(U E U^{\dagger}\right)}_{i, j}\right|
&&
\leq 
{\left(\Pi_{\nummodes, \numelec}\right)}_{i, i}
+ \sum_{j =1}^{\nummodes}
\left| {\left(U E U^{\dagger}\right)}_{i, j}\right|
\\
&\geq 
{\left(\Pi_{\nummodes, \numelec}\right)}_{i, i}
- \epsilon_{\mathrm{EV}},
&&
\leq
{\left(\Pi_{\nummodes, \numelec}\right)}_{i, i}
+ \epsilon_{\mathrm{EV}}.
\end{align}
and the unions of the first $\numelec$ and last $\nummodes - \numelec$ respectively satisfy
\begin{align}
\bigcup_{i=1}^{\numelec} K_i 
&\subseteq 
\left[
\min_{1 \leq i \leq \numelec} \left(F_{i, i} - r_i \right),
\max_{1 \leq i \leq \numelec} \left(F_{i, i} + r_i \right)
\right]
\subseteq
\left[1 - \epsilon_{\mathrm{EV}}, 1 + \epsilon_{\mathrm{EV}}\right]
,
\\
\bigcup_{i=\numelec+1}^{\nummodes} K_i 
&\subseteq 
\left[
\min_{\numelec + 1 \leq i \leq \nummodes} \left(F_{i, i} - r_i \right),
\max_{\numelec + 1 \leq i \leq \nummodes} \left(F_{i, i} + r_i \right)
\right]
\subseteq
\left[-\epsilon_{\mathrm{EV}}, \epsilon_{\mathrm{EV}}\right].
\end{align}
By supposition, $\epsilon_{\mathrm{EV}} < 1/2 < 1 - \epsilon_{\mathrm{EV}}$, and so these two regions are distinct.
Therefore, by the Gershgorin circle theorem, 
there are $\numelec$ eigenvalues of $\estimate{R}$ in ${[1-\epsilon_{\mathrm{EV}}, 1+\epsilon_{\mathrm{EV}}]}$ and $\nummodes - \numelec$ eigenvalues in $[-\epsilon_{\mathrm{EV}}, \epsilon_{\mathrm{EV}}]$.
Therefore, 
\begin{align}
\left\|D\right\| &=
\left\| \Lambda - \Pi_{\nummodes, \numelec}\right\|
\leq \epsilon_{\mathrm{EV}}
=
\nummodes^3 \epsilon_{\mathrm{shdw}}.
\end{align}
We also have
\begin{align}
\norm{\estimate{E}}
&\leq
\sqrt{\sum_{i, j=1}^{\nummodes} {\left|\estimate{E}_{i, j}\right|}^2}
\leq
\sqrt{\nummodes^2 \epsilon_{\mathrm{shdw}}^2}
=\nummodes \epsilon_{\mathrm{shdw}},
\end{align}
where we used the fact that the Frobenius norm upper bounds the operator norm.
Finally,
\begin{align}
\\
\left\|
D - \estimate{E}
\right\|
&\leq
\left\|
D
\right\|
+
\left\|
\estimate{E}
\right\|
\leq  2\nummodes^3 \epsilon_{\mathrm{shdw}}.
\end{align}
\end{proof}
 
\printbibliography[heading=bibintoc]

\end{refsection}
\end{document}